\documentclass[11pt]{article}

\usepackage[letterpaper,margin=1.0in]{geometry}
\usepackage{amsmath, amssymb, amsthm, amsfonts}
\usepackage{thmtools,thm-restate}
\usepackage{bbm}

\usepackage{ifthen}

\usepackage{cite}
\usepackage{appendix}
\usepackage{graphicx}
\usepackage{xcolor}
\usepackage{algorithm}
\usepackage[noend]{algpseudocode}

\usepackage{dsfont}

\usepackage[textsize=tiny]{todonotes}

\usepackage{framed}
\usepackage[framemethod=tikz]{mdframed}
\usepackage[bottom]{footmisc}
\usepackage{enumitem}
\setitemize{noitemsep,topsep=3pt,parsep=3pt,partopsep=3pt}
\usepackage[font=small]{caption}
\usepackage{xspace}

\definecolor{darkgreen}{rgb}{0,0.5,0}
\definecolor{darkblue}{rgb}{0,0,0.6}
\usepackage{hyperref}
\hypersetup{
    unicode=false,          
    colorlinks=true,        
    linkcolor=darkblue,          
    citecolor=darkgreen,        
    filecolor=magenta,      
    urlcolor=cyan           
}
\usepackage[capitalize, nameinlink]{cleveref}
\Crefname{theorem}{Theorem}{Theorems}
\Crefname{lemma}{Lemma}{Lemmas}
\Crefname{claim}{Claim}{Claims}
\Crefname{remark}{Remark}{Remarks}
\Crefname{observation}{Observation}{Observations}

\newtheorem{theorem}{Theorem}[section]
\newtheorem{lemma}[theorem]{Lemma}
\newtheorem{meta-theorem}[theorem]{Meta-Theorem}
\newtheorem{claim}[theorem]{Claim}

\newtheorem{corollary}[theorem]{Corollary}

\newtheorem{definition}{Definition}[section]

\algnewcommand\algorithmicswitch{\textbf{switch}}
\algnewcommand\algorithmiccase{\textbf{case}}

\algdef{SE}[SWITCH]{Switch}{EndSwitch}[1]{\algorithmicswitch\ #1\ \algorithmicdo}{\algorithmicend\ \algorithmicswitch}%
\algdef{SE}[CASE]{Case}{EndCase}[1]{\algorithmiccase\ #1}{\algorithmicend\ \algorithmiccase}%
\algtext*{EndSwitch}%
\algtext*{EndCase}%

\newcommand{\eps}{\varepsilon}

\newcommand{\CONGEST}{$\mathsf{CONGEST}$\xspace}
\newcommand{\LOCAL}{$\mathsf{LOCAL}$\xspace}

\newcommand{\local}{\LOCAL}

\newcommand{\poly}{\operatorname{\text{{\rm poly}}}}

\renewcommand{\tilde}{\widetilde}

\DeclareMathOperator{\E}{\mathbb{E}}
\newcommand{\Labels}{\Sigma}

\newcommand{\utility}{\mathbf{u}}
\newcommand{\cost}{\mathbf{c}}

\newcommand{\fC}{\mathcal{C}}

\newcommand{\fE}{\mathcal{E}}

\newcommand{\FullOrShort}{full}

\ifthenelse{\equal{\FullOrShort}{full}}{
	
  \newcommand{\fullOnly}[1]{#1}
  \newcommand{\shortOnly}[1]{}
	\renewcommand{\baselinestretch}{1.00}
  
  }{
    \renewcommand{\baselinestretch}{1.00}
    \newcommand{\fullOnly}[1]{}
    \newcommand{\IncludePictures}[1]{}
   
  }

\begin{document}
\date{}
\title{Faster Deterministic Distributed MIS and Approximate Matching}
\author{
  Mohsen Ghaffari \\
  \small{MIT}\\
  \small{ghaffari@mit.edu}
  \and
  Christoph Grunau \\
  \small{ETH Zurich}\\
  \small{cgrunau@inf.ethz.ch}
}
\maketitle

\begin{abstract} 
We present an $\tilde{O}(\log^2 n)$ round deterministic distributed algorithm for the maximal independent set problem. By known reductions, this round complexity extends also to maximal matching, $\Delta+1$ vertex coloring, and $2\Delta-1$ edge coloring. These four problems are among the most central problems in distributed graph algorithms and have been studied extensively for the past four decades. This improved round complexity comes closer to the $\tilde{\Omega}(\log n)$ lower bound of maximal independent set and maximal matching [Balliu et al. FOCS '19]. The previous best known deterministic complexity for all of these problems was $\Theta(\log^3 n)$. Via the shattering technique, the improvement permeates also to the corresponding randomized complexities, e.g., the new randomized complexity of $\Delta+1$ vertex coloring is now $\tilde{O}(\log^2\log n)$ rounds.

\medskip
Our approach is a novel combination of the previously known (and seemingly orthogonal) two methods for developing fast deterministic algorithms for these problems, namely \textit{global derandomization via network decomposition} (see e.g., [Rozhon, Ghaffari STOC'20; Ghaffari, Grunau, Rozhon SODA'21; Ghaffari et al. SODA'23]) and \textit{local rounding of fractional solutions} (see e.g., [Fischer DISC'17; Harris FOCS’19; Fischer, Ghaffari, Kuhn FOCS’17; Ghaffari, Kuhn FOCS'21; Faour et al. SODA'23]). We
consider a relaxation of the classic network decomposition concept, where instead of requiring the clusters in the same block to be non-adjacent, we allow each node to have a small number of neighboring clusters. We also show a deterministic algorithm that computes this relaxed decomposition faster than standard decompositions. We then use this relaxed decomposition to significantly improve the integrality of certain fractional solutions, before handing them to the local rounding procedure that now has to do fewer rounding steps.
\end{abstract}

\setcounter{page}{0}
\thispagestyle{empty}


\newpage
{   
    \smallskip
    \hypersetup{linkcolor=blue}
    \tableofcontents
    \setcounter{page}{0}
    \thispagestyle{empty}
}
\newpage

\section{Introduction}
In this paper, we present a faster deterministic distributed algorithm for the Maximal Independent Set (MIS) problem, which is one of the most central problems in distributed graph algorithms and has been studied extensively for the past four decades. This improvement has implications for several other problems. Our main novelty is a technique that combines the two previously known general approaches, which seemed unrelated and incompatible hitherto. We are hopeful that this new technique finds applications in a wider range of problems. We next set up the context, and then discuss our results and approach.

\paragraph{Distributed Model:} We work with the standard synchronous message-passing model of distributed computing, often referred to as the \local model, due to Linial~\cite{linial1987LOCAL}. The network is abstracted as an $n$-node undirected graph $G=(V, E)$, where each node represents one processor and a link between two nodes indicates that those two processors can communicate directly. Each processor has a unique $b$-bit identifier, where we typically assume $b=O(\log n)$. Initially, nodes do not know the topology of the network $G$, except for potentially knowing some global parameters such as a polynomial upper bound on $n$. Computations and communications take place in synchronous rounds. Per round, after doing arbitrary computations on the data that it holds, each process/node can send one message to each of its neighbors. In the \local model, the message sizes are not bounded. The model variant where message sizes are limited to $O(\log n)$-bits is known as the \CONGEST model\cite{peleg00}. The messages sent in a round are delivered before the end of that round. At the end of the computation, each node should know its own output, e.g., in the MIS problem, each node should know whether it is in the computed maximal independent set or not.

\subsection{State of the Art}\label{subsec:state-of-the-art}
\paragraph{Randomized algorithms, and the pursuit of deterministic algorithms.} In the 1980s, Luby\cite{luby86} and Alon, Babai, and Itai\cite{alon86} presented a simple and elegant randomized distributed algorithm that computes an MIS in $O(\log n)$ rounds, with high probability\footnote{As standard, the phrase \textit{with high probability} indicates that an event happens with probability at least $1-1/n$.}. Due to known reductions, this MIS algorithm led to $O(\log n)$ round randomized algorithms for many other key graph problems, including maximal matching, $\Delta+1$ vertex coloring, and $(2\Delta-1)$ edge coloring. These problems are often listed as the four fundamental symmetry-breaking problems in distributed graph algorithms and have a wide range of applications. The $O(\log n)$-round randomized algorithm naturally led the researchers to seek a deterministic distributed algorithm with the same round complexity. In his celebrated work~\cite{linial1987LOCAL, linial92}, Linial asked ``\textit{can it [MIS] always be found in polylogarithmic time [deterministically]?}''
He even added that ``\textit{getting a deterministic polylog-time algorithm for MIS seems hard.}'' Since then, this became known as Linial's MIS question and turned into one of the research foci in distributed graph algorithms. 

\paragraph{The two approaches of deterministic algorithms.} Linial's MIS problem remained open for nearly three decades. During this time, two different general approaches were developed and pursued. The first approach relies on \textit{global computations/derandomization via network decompositions}. The second approach is based on \textit{local rounding of certain fractional solutions}. Over the past couple of years, both approaches came to fruition and led to two completely independent polylogarithmic time deterministic distributed algorithms for MIS~\cite{rozhonghaffari20, faour2022local}. The round complexities are still $\tilde{\Omega}(\log^3 n)$, which is somewhat far from the randomized $O(\log n)$ complexity~\cite{luby86, alon86}. We next discuss each approach separately.

\paragraph{(I) Global computation/derandomization via network decomposition.} The first approach is based on the concept of network decomposition, which was introduced by Awerbuch, Goldberg, Luby, and Plotkin~\cite{awerbuch89} as the key tool in developing deterministic distributed algorithms for MIS and other symmetry-breaking problems. A $(c, d)$ network decomposition is a partition of the vertex set into $c$ disjoint parts, each known as a \textit{block}, such that in the subgraph induced by the nodes in each block, each connected component has a diameter of at most $d$. Said differently, each block consists of non-adjacent clusters, each of diameter\footnote{If the diameter is measured in the induced subgraph, this is known as a strong-diameter network decomposition. If the distance is measured in the original graph, this is known as a weak-diameter network decomposition.} at most $d$. Given a $(c, d)$ network decomposition, it is easy to compute an MIS in $O(cd)$ rounds: The rounds are organized in $c$ iterations, each consisting of $O(d)$ rounds. In iteration $i$, we add to the output independent set an MIS of the nodes of block $i$, which can be computed easily in $O(d)$ rounds as each cluster in the block has a diameter of at most $d$. We then remove any node in any other block that has a neighbor in this independent set and then move to the next iteration. Any $n$-node graph has a $(c, d)$ network decomposition for $c=d=O(\log n)$~\cite{awerbuch85}. Awerbuch et al.\cite{awerbuch89} gave a deterministic distributed algorithm that computed a $(c, d)$ decomposition in $t$ rounds for $c=d=t=2^{O(\sqrt{\log n  \log\log n})}$. These bounds were improved to $2^{O(\sqrt{\log n})}$ by Panconesi and Srinivasan~\cite{panconesi-srinivasan}. 
By a technique of Awerbuch, Berger, Cowen, and Peleg~\cite{awerbuch96}, one can transform any $t$-round $(c, d)$ decomposition algorithm to a $(t+cd) \cdot \poly(\log n)$-round $(\log n, \log n)$ network decomposition algorithm in the \local model.

The $2^{O(\sqrt{\log n})}$ bounds of Panconesi and Srinivasan remained the state of the art for over 25 years, until getting improved dramatically to $\poly(\log n)$ by a new decomposition algorithm of Rozhon and Ghaffari~\cite{rozhonghaffari20}. That gave the first polylogarithmic-time deterministic distributed algorithm for MIS, and thus the first positive resolution to Linial's MIS question. The algorithm was optimized in a follow-up work of Ghaffari, Grunau, and Rozhon~\cite{GGR20}, which brought down MIS's deterministic round complexity to $O(\log^5 n)$. See also \cite{chang2021strong,elkin2022deterministic} which obtain strong-diameter guarantees with small messages. In a very recent work, Ghaffari, Grunau, Haeupler, Ilchi, and Rozhon~\cite{ghaffari2023netdecomp} presented a completely different and faster method for computing network decompositions. Their algorithm computes a $(c, d)$ strong-diameter network decomposition for $c=O(\log n)$ and $d=O(\log n \cdot \log\log\log n)$ in $\tilde{O}(\log^{3} n)$ rounds, using $O(\log n)$-bit messages. This is the state of art decomposition in essentially all regards, and it provides an $\tilde{O}(\log^{3} n)$ round deterministic algorithm for MIS.

\paragraph{(II) Local rounding of fractional solutions.} The second approach is based on obtaining fractional solutions to certain relaxations of the problem and then, locally and gradually rounding these solutions into integral solutions. Unlike the network decomposition approach, which was obviously applicable to all the symmetry-breaking problems from the start, the applicability of the rounding approach appeared limited at first and gradually increased. It started first with only the maximal matching problem, then extended to $2\Delta-1$ edge coloring, then to $\Delta+1$ vertex coloring, and finally to the hardest of all, the MIS problem. 

The starting point is the work of Hanckowiak, Karonski, and Panconesi~\cite{hanckowiak01} who gave the first $\poly\log n$-time deterministic distributed algorithm for the maximal matching problem. Fischer\cite{fischer2020improved} rephrased their approach in a \textit{rounding} language and used this to improve the maximal matching complexity to $O(\log^2 \Delta \cdot \log n)$. 

Fischer's rounding was very specific to matching in graphs. Fischer, Ghaffari, and Kuhn\cite{FischerGK17} developed a different rounding method for matchings that extended to low-rank hypergraphs. By a reduction that they provided from $(2\Delta-1)$-edge coloring in graphs to maximal matching in hypergraphs of rank $3$, this led to a $\poly\log n$-time deterministic algorithm for $(2\Delta-1)$-edge coloring, hence putting the second problem in the $\poly(\log n)$ regime. Harris\cite{harris2019distributed} improved the complexity to $\tilde{O}(\log^2 \Delta \cdot \log n)$. 

The above local rounding approaches appeared limited to matching in graphs or hypergraphs, until a work of Ghaffari and Kuhn~\cite{GhaffariK21}. They developed an efficient rounding method for $\Delta+1$ coloring (this was shortly after the first polylogarithmic-time network decomposition result~\cite{rozhonghaffari20}). This was by reexamining the analysis of the natural randomized coloring algorithm, seeing it as a fractional/probabilistic assignment of colors to the nodes which has a small bound on the expected number of monochromatic edges, and then gradually rounding the fractional assignments while approximately maintaining the upper bound on the expected number of monochromatic edges.

Finally, in a very recent work, Faour, Ghaffari, Grunau, Kuhn, and Rozhon~\cite{faour2022local} found a significant generalization of the above local roundings. They presented a unified method that can provide local rounding for any problem whose randomized solution analysis relies on only pairwise independence. This led to algorithms with round complexity $O(\log^2 \Delta \log n)$ for MIS and thus also for the other three problems mentioned above, and in a unified way.

\subsection{Our Results}
\paragraph{Unifying the two approaches and faster deterministic MIS.} As discussed above, prior to the present paper, the best known deterministic complexity of MIS (and indeed any of the other three key symmetry breaking problems) remained $\Omega(\log^3 n)$\cite{ghaffari2023netdecomp, faour2022local}. This is somewhat far from the randomized $O(\log n)$ round complexity~\cite{luby86, alon86}. Furthermore, the two approaches seemed unrelated and incompatible. In this paper, we present a method to combine the two general approaches, which allows us to achieve a deterministic round complexity of $\tilde{O}(\log^2 n)$ for MIS.

\begin{restatable}{theorem}{mis}
    \label{thm:mis}
    There is a deterministic distributed algorithm, in the \local model, that computes a maximal independent set (MIS) in $\tilde{O}(\log^2 n)$ rounds. This also implies that there are \local-model deterministic distributed algorithms with $\tilde{O}(\log^2 n)$ rounds complexity also for maximal matching, $(deg+1)$-list vertex coloring, and $(2deg-1)$-list edge coloring\footnote{In the $(deg+1)$-list vertex coloring, each node $v$ having a prescribed list $L_{v}$ of colors, of size $|L_{v}|\geq deg(v)+1$ from which it should choose its color. In the $(2deg-1)$-list edge coloring, each edge $e=\{v, u\}$ has a prescribed list $L_{e}$ of colors, of size $|L_{e}|\geq deg(v)+deg(u)-1$ from which it should choose its color. Both problems reduce to MIS by a reduction of Luby~\cite{luby86, linial92}.}.
\end{restatable}

This comes closer to the $\tilde{\Omega}(\log n)$ lower bound, due to Balliu et al.\cite{balliu2019LB}, which holds for deterministic distributed algorithms that compute maximal independent set or maximal matching. 

\paragraph{Faster approximate maximum matching.} Our method also leads to a faster deterministic algorithm for computing a constant approximation of the maximum matching, in $\tilde{O}(\log^{4/3} n)$ rounds. By adding to this the approximation booster of Fischer, Mitrovic, and Uitto~\cite{fischer2022matching}, we can improve the approximation to $(1+\epsilon)$ for any desirably small constant $\epsilon>0$ with no asymptotic round complexity loss. The previous best round complexity for constant approximation of the maximum matching was $O(\log^2 \Delta)$~\cite{fischer2020improved, faour2022local}.
\begin{restatable}{theorem}{approximatematching}
    \label{thm:matching_approximate}
    There is a deterministic distributed algorithm, in the \local model, that computes a $\Theta(1)$-approximate maximum matching in $\tilde{O}(\log^{4/3} n)$ rounds. This can be boosted to an algorithm that computes a $(1+\epsilon)$ approximation of maximum matching in $\tilde{O}(\log^{4/3} n) \cdot \poly(1/\epsilon)$ rounds.
\end{restatable}
\shortOnly{Due to space limitations, the proof of \Cref{thm:matching_approximate} is deferred to the full version of this paper.}
\paragraph{Corollaries for randomized coloring.} Our faster deterministic algorithm can be plugged in into the shattering framework of randomized algorithms\cite{barenboim12_decomp, chang2018optimal} and improves the randomized complexity for $\Delta+1$ vertex coloring and $2\Delta-1$ edge coloring.  


\begin{corollary}[\Cref{{thm:mis}}+\cite{chang2018optimal}]
There is a randomized distributed algorithm, in the \local model, that computes a $\Delta+1$ vertex coloring in $\tilde{O}((\log\log n)^2)$ rounds. The same holds also for the $2\Delta-1$ edge coloring problem.
\end{corollary}


\subsection{Other Related Work}\label{subsec:related}
We discussed in \Cref{subsec:state-of-the-art} deterministic constructions for network decomposition and the generic method of using network decomposition for symmetry-breaking problems, e.g., MIS. Let us add here three side comments and mention the related work.

(1) Better randomized constructions for network decomposition have been known. In particular, the work of Linial and Saks~\cite{linial93} presented an $O(\log^2 n)$ round randomized algorithm for computing decompositions with $O(\log n)$ colors and $O(\log n)$ weak diameter. Elkin and Neiman~\cite{elkin16_decomp} imported a parallel algorithm of Miller, Peng and Xu~\cite{miller2013parallel} into the distributed setting and obtained an $O(\log^2 n)$ round randomized algorithm for computing decompositions with $O(\log n)$ colors and $O(\log n)$ strong diameter.

(2) The deterministic MIS method described in \Cref{subsec:state-of-the-art} for using network decompositions in computing symmetry-breaking problems such as MIS would require large messages, as it gathers the topology of each cluster in a center. For MIS, one can work with $O(\log n)$-bit messages, by using a derandomization method of Censor-Hillel, Parter, and Schwartzman~\cite{censor2017derandomizing}, and that gives a deterministic algorithm with $O(\log n)$-bit messages and round complexity $O(cd) \cdot \poly(\log n)$. Similar approaches have been presented for other problems, see e.g., Bamberger et al.~\cite{bamberger2020coloring} for results on $\Delta+1$ coloring. 

(3) Ghaffari, Harris, and Kuhn~\cite{ghaffari2018derandomizing} and Ghaffari, Kuhn and Maus~\cite{ghaffari2017complexity} showed that one can get a general derandomization method for the \local model using network decompositions. This method transforms any $\poly(\log n)$-round randomized algorithm for any problem whose solution can be checked in $\poly(\log n)$ rounds into a deterministic algorithm with round complexity $O(cd+t) \cdot \poly(\log n)$, assuming we have a deterministic $(c, d)$ network decomposition algorithm with round complexity $t$. 

\section{Our Approach}
To discuss our approach and put it in the context of prior work, let us use the problem of maximal matching. This is a special case of the maximal independent set problem (on line graphs), but it will suffice for explaining most of the key ideas in our overall approach.

\subsection{The previous methods and their barriers} As mentioned before, the previous two approaches for deterministic algorithms both get stuck at a round complexity of $\tilde{O}(\log^3 n)$. Let us briefly overview the approaches and their barriers here. This discussion will help us explain how we obtain a faster algorithm.

\paragraph{Local rounding of fractional solutions.} For maximal matching, Fischer~\cite{fischer2020improved} and Faour et al.~\cite{faour2022local} provide two different local rounding algorithms, both achieving an $O(\log^2 \Delta \log n)$ complexity. Both round complexities are stuck at essentially the same barrier, so we discuss only Fischer's approach. It is easy to deterministically obtain a constant-approximate fractional matching in $O(\log \Delta)$ rounds\footnote{Start with a fractional assignment of $x_e=1/\Delta$ to all edges. Then, we have $\log \Delta$ iterations. Per iteration, for each edge $e$ such that each endpoint $v$ of it has $\sum_{e' \ni v} x_{e'}\leq 1/2$, set $x_e\gets 2x_e.$}~\cite{fischer2020improved}. The real challenge for deterministic algorithms is in rounding this fractional matching into an integral matching. Fischer~\cite{fischer2020improved} gradually improves the integrality of the fractional solution from $1/\Delta$ to $1$, in $\log \Delta$ gradual rounding steps, each time doubling the minimum edge value. Each doubling step in the algorithm takes $\Theta(1/\eps) = \Theta(\log \Delta)$ rounds, as we need to ensure that there is a loss of at most $\epsilon=1/(2\log \Delta)$ fraction in the matching size, so that the total loss over all the $\log \Delta$ iterations is at most a constant. One can see that, for an arbitrary fractional matching, this $\Theta(1/\eps)$ complexity is the best possible\footnote{Here is an intuitive and informal explanation: Consider a long path of edges, each with fractional value $1/2$. We would like to raise some edge values to $1$ while dropping others to $0$, and while ensuring that the total size is still a $1-\eps$ factor of the previous fractional size. Intuitively, starting from one edge, the edges should be alternating in $0$ and $1$ value until we go at least $1/\eps$ far. Otherwise, we have had more than an $\eps$ relative loss in this part. This strict alternation for $1/\eps$ distance requires $\Omega(1/\eps)$ rounds.}. Considering the $\log \Delta$ gradual rounding steps, this yields a constant-approximate integral matching in $O(\log^2 \Delta)$ rounds. We emphasize that this runtime is quadratic in the number of rounding steps, as we need to ensure that the relative loss per step is less than the inverse of the number of steps. Via $O(\log n)$ repetitions of constant-approximate maximum matching, each time on the subgraph induced by vertices that remained unmatched, one gets a maximal matching with complexity $O(\log^2 \Delta \cdot \log n)$. 

\paragraph{Global computation via network decomposition.} The state-of-the-art network decomposition algorithm~\cite{ghaffari2023netdecomp} computes a decomposition with $O(\log n)$ vertex-disjoint blocks, where the subgraph induced by each block consists of connected components of diameter $\tilde{O}(\log n)$, in $\tilde{O}(\log^3 n)$ rounds.  Given this, it is easy to compute maximal matching or MIS in $\tilde{O}(\log^2 n)$ rounds. But the construction time is the bottleneck in these applications. 

Let us revisit this decomposition-based approach for the constant-approximate maximum matching problem. This discussion will help us later to explain our approach. As discussed above, $O(\log n)$ repetitions of constant-approximate maximum matching give a maximal matching.  Consider a constant-approximate \textit{fractional} maximum matching. As mentioned above, we know how to compute this easily in $O(\log \Delta)$ rounds. The task is to turn this into an integral matching of the same size, up to a constant factor. The network decomposition algorithm~\cite{ghaffari2023netdecomp} computes the blocks one by one, each in $\tilde{O}(\log^2 n)$ rounds, such that each block clusters a constant fraction of the remaining nodes into non-adjacent clusters, each of diameter $\tilde{O}(\log n)$. We can adjust the block construction in such a way that, instead of a constant fraction of nodes getting clustered, we ensure that at least a constant fraction of the weight of the fractional matching is on the edges inside the clusters. It is easy to see that randomized rounding can turn (each cluster's) fractional matching into an integral matching of the same size up to a constant factor. So, to deterministically do the rounding, we can simply gather each cluster's information (and constraints) into the cluster center and compute the integral matching. This takes $\tilde{O}(\log n)$ rounds, as the cluster diameter is $\tilde{O}(\log n)$. It is worth noting that the current method for building each block of the decomposition requires $\tilde{O}(\log^2 n)$ rounds, as it has $\tilde{O}(\log n)$ steps of randomness fixing, each involving coordination along $\tilde{O}(\log n)$-hop distances. 

\subsection{Low-diameter low-degree clusterings}
In our approach, we work with a relaxed variant of the standard network decompositions. In particular, we work with clusterings that have a lower diameter than the standard $\Theta(\log n)$ bound. This opens the road for faster deterministic constructions. However, it comes at the disadvantage that we do not have only $\log n$ blocks, each consisting of non-adjacent clusters. We instead maintain a small bound on the maximum number of clusters that are adjacent to each particular node. To state the formal decomposition, let us first present the formal definitions.

\begin{definition}\label{def:clusterings}[Cluster, Clustering, Partition, Cluster Degree of a Vertex, Diameter of a Cluster]
A subset $C \subseteq V$ is called a cluster. The (strong)-diameter of a cluster is defined as $\max_{u,v \in C} d_{G[C]}(u,v)$.
A clustering $\fC$ is a set of disjoint clusters. We refer to $\fC$ as a partition if $\bigcup_{C \in \fC} C = V$. The diameter of a clustering $\fC$ is defined as the maximum diameter of all its clusters.
For a node $u \in V$ and a clustering $\fC$, we define the degree of $u$ with respect to $\fC$ as $deg_\fC(u) = |\{C \in \fC \colon d(C,u) \leq 1\}|$. We sometimes refer to $\max_{u \in V} \deg_\fC(u)$ as the degree of this clustering.
\end{definition}

We show an algorithm that computes $O(\alpha)$-diameter $(2^{\tilde{O}(\log n)/\alpha})$-degree clusterings in $\tilde{O}(\alpha^2 \log n)$ rounds of the \local model. 
\begin{restatable}{theorem}{clusteringall}
   \label{thm:clustering_all}
   Let $\alpha \leq \log n$ be an arbitrary value.
   There exists a deterministic distributed algorithm that in $\tilde{O}(\alpha^2 \log n)$ rounds of the \local model computes a partition $\fC$ with
   \begin{enumerate}
       \item diameter $O(\alpha)$ and
       \item $\max_{u \in V} \deg_\fC(u) = 2^{\tilde{O}(\log n)/\alpha}$, i.e., each node has neighbors in at most $2^{\tilde{O}(\log n)/\alpha}$ clusters.
   \end{enumerate}
\end{restatable}

We note that the existence of such a clustering/partition, and even efficient randomized constructions for it, follow from known randomized methods~\cite{miller2013parallel}. We provide an efficient deterministic distributed construction here, as we outline in \Cref{subsec:construction}. Furthermore, we show that this relaxed notion is still useful for applications in the maximal independent set and maximal matching problems, as we next outline in \Cref{subsec:interplay}. 

As a side comment, we note that for some applications (e.g., approximate-matching), it suffices to have a weaker version of \Cref{thm:clustering_all} where only a constant fraction of nodes (or a constant fraction of the weight of the nodes, according to some given weights) is clustered. For that version, we have a faster algorithm that runs in  $\tilde{O}(\alpha \log n)$ rounds of the \local model. The precise statement is as follows.

\begin{restatable}{theorem}{clusteringconstant}
    \label{thm:clustering_constant}
    Let $\alpha  \leq \log n$ be an arbitrary value and assume that each node $u \in V$ is equipped with a value $x_u \in \left[\frac{1}{n},1 \right]$.
    There exists a deterministic distributed algorithm that in  $\tilde{O}(\alpha \log n)$ rounds of the \local model computes a partition $\fC$ with
    
    \begin{enumerate}
        \item diameter $O(\alpha)$ and
        \item let $V^{good} = \{u \in V \colon deg_\fC(u) = 2^{\tilde{O}(\log n)/\alpha}\}$, then we have $\sum_{u \in V^{good}} x_u \geq 0.9\sum_{u \in V} x_u$.
    \end{enumerate}
\end{restatable}

\subsection{Interplay between global computations via clustering and local rounding}
\label{subsec:interplay}
Let us now revisit the maximal matching problem. Here, we provide an intuition of how we solve this problem in $\tilde{O}(\log^2 n)$ rounds, by combining the two approaches. First, we invoke \Cref{thm:clustering_all} for $\alpha=\sqrt{\log n}$ and compute a partition $\fC$ into clusters of diameter $O(\sqrt{\log n})$ such that each node has $d=2^{\tilde{O}(\sqrt{\log n})}$ neighboring clusters, in $O(\log^2 n)$ rounds. This vertex partition also induces an edge partition $E = \sqcup_{C \in \fC} E_C$ where $E_C$ contains all edges whose endpoint with the larger identifier is in $C$. Now, we use this clustering to gradually compute the maximal matching in $O(\log n)$ iterations. In each iteration, we first compute a constant-approximate fractional matching in the graph induced by the unmatched vertices in $O(\log \Delta)=O(\log n)$ rounds~\cite{fischer2020improved}. Then, the key part of the iteration is to turn this into a constant-approximate integral matching in $\tilde{O}(\log n)$ additional rounds. For that, we use the clustering to quickly perform a significant rounding of the fractional matching. In particular, the value assigned to a given edge $e \in E_C$ after the rounding is fully determined by the fractional values assigned to the edges in $E_C$ before the rounding. Thus, this initial rounding can be computed in just $O(\sqrt{\log n})$ rounds. Afterward, we will resort to local rounding which now needs only $\tilde{O}(\sqrt{\log n})$ doubling steps and thus can be performed in $\tilde{O}(\log n)$ rounds. Let us elaborate.

First, we divide all the fractional matching values by a $4$ factor, so that now each node $v$ has $\sum_{e' \ni v} x_{e'}\leq 1/4$. Then, for the rounding, for each node $v$ and cluster $C$, we allow the summation of the edges in $E_C$ incident to $v$ to increase from the original value by a multiplicative factor of $2$ and, on top of that, an additive increase of  $1/(2d)=1/2^{\tilde{\Theta}(\sqrt{\log n})}$. Hence, the multiplicative deviations bring the summation to at most $1/2$ and the additive deviations, over all the $d$ neighboring clusters, raise the total sum to at most $1$. Thus, this would still ensure that we have a valid fractional matching. How does each cluster compute such a rounded fractional solution that satisfies this constraint for all possible neighboring nodes $v$?

For each cluster $\mathcal{C}$, such a solution exists by a simple probabilistic rounding argument. Let us provide an informal explanation. Suppose we keep the fractional values that exceed $1/\Theta(d\log n)$ intact. For any edge that has $x_e \leq 1/\Theta(d\log n)$, let us round it probabilistically: with probability $p=x_e \cdot \Theta(d\log n)$, we set it to $1/\Theta(d\log n)$, and with the rest of the probability, we set it to zero. One can see by a standard Chernoff bound that, w.h.p., for the edges in $E_C$ incident on node $v$, the new summation of fractional values is within a $2$ factor of the old summation, modulo an additive error of at most $1/(2d)$. With this rounding, now the smallest fractional value is at least $1/\Theta(d\log n) =1/2^{\tilde{\Theta}(\sqrt{\log n})}$. Furthermore, with a reasonably high probability, the size of the fractional matching is preserved up to a constant factor inside the cluster\footnote{We defer the precise discussions of the details to \Cref{sec:MMandMIS}.}. To deterministically find the solution proven to exist by this probabilistic argument, it suffices for the cluster $\mathcal{C}$ to gather the current fractional values of the edges in $E_C$. Then, it can find such a fractional solution in a centralized fashion. Since the diameter of the cluster is $O(\sqrt{\log n})$, this can be done easily in $O(\log n)$ rounds of the \local model. 

Once each cluster does this rounding (and note that all can be done in parallel), we have computed a fractional matching with integrality $1/2^{\tilde{\Theta}(\sqrt{\log n})}$ whose size is at least a constant fraction of the fractional matching with which we started. Hence, it suffices to apply only $\tilde{O}(\sqrt{\log n})$ steps of local rounding of \cite{fischer2020improved} on this fractional matching, and that takes $\tilde{O}(\sqrt{\log n})^2 = \tilde{O}(\log n)$ rounds. Therefore, we now have a method to compute a constant-approximate integral matching in $\tilde{O}(\log n)$ rounds. This is after the initial $\tilde{O}(\log^2 n)$ time spent for computing the clustering, but that clustering is computed only once. With $O(\log n)$ repetitions of this constant-approximate integral matching in $\tilde{O}(\log n)$ rounds, each time removing the matched nodes from the future matching computations, we get a maximal matching in $\tilde{O}(\log^2 n)$ rounds.

\subsection{Construction of the clustering, and hitting set with pipelining}
\label{subsec:construction}
Our low-diameter clustering algorithms are similar in spirit to Ghaffari et al.~\cite{ghaffari2023netdecomp}. We start by briefly overviewing their approach.

\paragraph{Review of the clustering algorithm of \cite{ghaffari2023netdecomp}.}

Recently, Ghaffari et al. \cite{ghaffari2023netdecomp} obtained a network decomposition with $O(\log n)$ colors and diameter $O(\log n \log \log \log n)$.
One key step towards that result is an algorithm that, in $\tilde{O}(\log^2 n)$ rounds, computes a partition of diameter $O(\log n)$ such that a large constant fraction of the vertices have a clustering degree of $O(\log \log n)$. Their algorithm to compute such a partition can be seen as a derandomization of the randomized algorithm of Miller, Peng and Xu \cite{miller2013parallel}. 
The randomized algorithm of Miller, Peng and Xu computes a partition of diameter $O(\log n)$ by first assigning each vertex a (random) delay between $0$ and $O(\log n)$. Then, each node $u$ gets clustered to the node $v$ minimizing $del(v) + d(v,u)$ where $del(v)$ is the delay assigned to node $v$.
The delays assigned to all of the vertices can be computed by repeatedly, for $\log n$ repetitions, subsampling all the remaining active nodes with probability $1/2$. Here, all the nodes would be active at the beginning, and in each repetition we keep each previously active node with probability $1/2$. With high probability, no active node remains after $O(\log n)$ subsampling steps. 
The deterministic algorithm of Ghaffari et al. derandomizes each of the $O(\log n)$ subsampling steps. To do so, they phrase each subsampling step as an instance of a certain hitting set variant. They show that the randomized algorithm produces a partition such that the expected cluster degree of each node is $O(\log \log n)$. The hitting set viewpoint in each subsampling step lends itself to pairwise analyses and opens the road for efficient derandomization, allowing one to ``sample" the vertices of each step in $\tilde{O}(\log n)$ rounds in a deterministic manner.

 \paragraph{Our Clustering Algorithms.} For our weaker clustering result, namely \Cref{thm:clustering_constant} which clusters only a constant fraction of the vertices (or their weights), we follow a similar general approach. One can see that, for any $\alpha \leq \log n$, by a simple parameter adjustment in \cite{miller2013parallel}, their randomized construction produces a partition with diameter $O(\alpha)$ and cluster degree $2^{O(\log n/\alpha)}$--essentially, it suffices to have $\alpha$ subsampling steps, each with probability $2^{-O(\log n/\alpha)}$. We can then follow a similar derandomization approach as \cite{ghaffari2023netdecomp} to turn this into a deterministic algorithm. Indeed, because of the parameter regime, we can work here with a even slightly simpler hitting set analyses, as we describe in \Cref{sec:clustering_constant}. This leads to clustering a constant fraction of the nodes (or a constant fraction of their weights) in $\tilde{O}(\alpha \log n)$ rounds. Intuitively, in contrast to the $\tilde{O}(\log^2 n)$ complexity of the clustering of \cite{ghaffari2023netdecomp}, the complexity here is $\tilde{O}(\alpha \log n)$ because the distance in the hitting set problem (to coordinate between the nodes to be hit and the active nodes) is reduced from $O(\log n)$ to $O(\alpha)$.
 
 For the stronger clustering algorithm that we have, namely \Cref{thm:clustering_all}, we want to cluster all vertices. Doing this by repetitions of \Cref{thm:clustering_constant}, which clusters only a constant fraction of nodes (or their weights), would incur a factor of $O(\log n)$ loss in the round complexity. That would bring the complexity of our clustering to $\tilde{O}(\alpha \log^2 n)$. \Cref{thm:clustering_all} however achieves a complexity of $\tilde{O}(\alpha^2 \log n)$. For that, we present a novel pipelining idea in the hitting set framework, as we outline next.

\paragraph{Efficient Hitting Set Derandomization with Pipelining}

Recently, Faour et al. \cite{faour2022local} gave a local rounding method that essentially can derandomize pairwise analyses. Consider some random process where each node gets sampled with probability $p$, pairwise independently. Then, assuming the pairwise analyses looks at only pairs of nodes that are nearby, their method allows one to derandomize the sampling in roughly $O(\log^2(1/p))$ rounds of the \local model. 

Let us provide a brief intuition: their algorithm views the probabilistic solutions as a $p$-integral fractional assignment, and gradually rounds it to an integral solution in $s=O(\log(1/p))$ doubling steps. The approach views the objectives analyzed via pairwise analyses as its target function, which should be approximately preserved. Each step uses a certain defective coloring and then decides the rounding of the nodes of each color simultaneously. This ensures that, with the exception of the small loss in the target function over edges that are colored monochromatically, the rest of the target function is preserved. The coloring is chosen such that the loss in each step is roughly a $1/(2s)$ fraction of the overall target function, so that the total loss over all $s=O(\log(1/p))$ steps is still below a constant fraction. This aspect of the parameterization is similar to the basic rounding of Fischer~\cite{fischer2020improved} for matchings, as discussed before.

Now, let us illustrate our pipelining idea by considering a simplified setup. Let $G = (U \sqcup V, E)$ be a bipartite graph. We assume that each node $u \in U$ has degree $\Delta$ and we define $n = |U|$. The goal is to compute a small subset $V' \subseteq V$ such that each node in $U$ is neighboring with at least one node in $V'$. Such a set $V'$ is referred to as a hitting set.

Consider the probabilistic approach where each node in $V$ joins $V'$ independently with probability $p = O(\log(n)/\Delta)$. The expected size of $V'$ is $O \left( \frac{\log(n) |V|}{\Delta} \right)$. Moreover, each node $u \in U$ is hit with probability $1 - (1-p)^{\Delta} \geq 1 - \frac{1}{n^{10}}$. Unfortunately, this analysis completely breaks down if we only assume pairwise independence.

A hitting set with the same guarantees can be computed using only pairwise independence as follows: First, let each node in $V$ join $V'$ with probability $\frac{1}{10 \Delta}$, pairwise independently. A simple pairwise analysis shows that each node has one neighbor in $V'$ with probability at least $\frac{1}{100}$. Therefore, repeating this subsampling step $O(\log n)$ times, and adding every node to $V'$ that was sampled at least once, results in each node being hit with high probability. Moreover, the expected size of $V'$ is $O(\frac{\log(n)|V|}{\Delta})$. 

The method of Faour et al. \cite{faour2022local} allows to directly derandomize the pairwise analysis. In particular, it allows to compute in $O(\log^2 \Delta)$ rounds a subset $V' \subseteq V$ of size $O(|V|/\Delta)$ such that at least $\frac{n}{100}$ vertices of $U$ have a neighbor in $V$.
Hence, in $O(\log n \log^2 \Delta)$ rounds, one can deterministically compute a hitting set of size $O(\log(n)|V|/\Delta)$.
This is not efficient enough for us to achieve \Cref{thm:clustering_constant}, especially for large $\Delta \gg \log n$.

We give a method that, in just $\tilde{O}(\log n)$ rounds, deterministically computes a slightly larger hitting set of size $O(|V|/\Delta^{1/\log \log n})$. Let $k = O(\log \Delta/\log \log n)$ be the largest integer such that $(100 \log n)^k \leq \Delta$. For $j \in [0,k]$, let $\deg_j = (100 \log n)^{k-j}$. The algorithm computes a sequence of sets $V:= V_0 \supseteq V_1 \supseteq \ldots \supseteq V_k$ such that each node in $U$ has at least $\deg_j$ neighbors in $V_j$. In particular, each node in $U$ has at least $\deg_k \geq 1$ neighbors in $V_k$ and therefore $V_k$ is a hitting set. The algorithm also ensures that $|V_j| \leq |V|/2^j$ and therefore $|V_k| = O(|V|/\Delta^{1/\log \log n})$, as promised. In the randomized world, one could obtain $V_{j+1}$ from $V_j$ by subsampling each vertex with probability $1/2$. Then, given that each node in $U$ neighbors at least $deg_j$ nodes in $V_j$, one can show that with high probability each node in $U$ neighbors at least $deg_{j+1}$ nodes in $V_{j+1}$.
One can also derandomize this with round complexity $\tilde{O}(\log n)$, which then would result in a round complexity of $\tilde{O}(\log n \log \Delta)$.
We give a method that pipelines the computation of the sets, in the sense that it starts the computation of $V_{j+1}$ before we have finished computing $V_j$. The actual ingredients of this pipelining are more involved and are thus deferred to the technical section. We hope that similar pipelining ideas might find applications in other similar contexts.

\section{Preliminaries}

\subsection{Definitions and Basic Inequalities}
\paragraph{Notation and Basic Definitions.} We use $n$ to denote the number of nodes and $N$ to denote a polynomial upper bound on the number of nodes. For two integers $a, b$ where $a\leq b$, we define $[a, b] = \{a, a+1, \dots, b\}$. For an event $\mathcal{E}$, we define the indicator variable $I(\mathcal{E})$ to be equal to $1$ if $\mathcal{E}$ happens and $0$ otherwise. 

Given a graph $G=(V, E)$, a fractional matching is an assignment of a value $x_e\in [0, 1]$ to each $e\in E$ such that for every vertex $v\in V$, we have $\sum_{e\ni v} x_e \leq 1$. The size of this fractional matching is $\sum_{e} x_e$. We call the fractional matching $1/Q$-integral if each edge $e$ satisfies either $x_e=0$ or $x_e \geq 1/Q$. We use this  $1/Q$-integrality term more generally also for other fractional assignments.

We next state some concentration inequalities that we will use throughout the paper.
\begin{theorem}[Chernoff Bound]
    \label{thm:chernoff}
    Let $X := \sum_{i \in [1,n]} X_i$, where $X_i$, $i \in [1,n]$ are independently distributed and $0 \leq X_i \leq 1$. Then, for a given $\delta > 0$, we have
    \[Pr[|X - \E[X]| \geq \delta\E[X]] \leq 2e^{-\frac{\min(\delta,\delta^2)\E[X]}{3}}.\]
\end{theorem}

\begin{corollary}[Chernoff Bound variant]
    \label{cor:chernoff}
    Let $Y := \sum_{i \in [1,n]} Y_i$, where $Y_i$, $i \in [1,n]$ are independently distributed and $0 \leq Y_i \leq b$ for some $b > 0$. Then, for a given $t \geq 0.5E[Y]$, it holds that
    \[Pr[|Y-\E[Y]| \geq t] \leq 2e^{-\frac{t}{6b}}.\]

\end{corollary}
\fullOnly{
\begin{proof}
    For every $i \in [1,n]$, let $X_i = Y_i/b$ and $X = \sum_{i=1}^n X_i$. Moreover, let $\delta := \frac{t}{\E[Y]} \geq \frac{1}{2}$
    Then, using \cref{thm:chernoff}, we get
    \fullOnly{\[Pr[|Y-\E[Y]| \geq t] = Pr[|X-\E[X]| \geq \delta \E[X]] \leq 2e^{-\frac{\min(\delta,\delta^2)\E[X]}{3}} \leq 2e^{-\frac{\delta\E[X]}{6}} = 2e^{-\frac{t}{6b}}.\]}
    \shortOnly{
    \begin{align*} 
    Pr[|Y-\E[Y]| \geq t] &= Pr[|X-\E[X]| \geq \delta \E[X]]  \\ &\leq 2e^{-\frac{\min(\delta,\delta^2)\E[X]}{3}}
    \leq 2e^{-\frac{\delta\E[X]}{6}} = 2e^{-\frac{t}{6b}}.
    \end{align*}
    }
\end{proof}
}

\subsection{The local rounding framework}\label{subsec:roundingFramework}
We use the local rounding framework of Faour et al.~\cite{faour2022local}. Their rounding framework works via computing a particular weighted defective coloring of the vertices, which allows the vertices of the same color to round their values simultaneously, with a limited loss in some objective functions that can be written as a summation of functions, each of which depends on only two nearby nodes. Next, we provide a related definition and then state their black-box local rounding lemma.

\begin{restatable}{definition}{DEFutilcost} 
\label{def:utility_cost}
[Pairwise Utility and Cost Functions]
Let $G = (V_G,E_G)$ be a graph. For any label assignment $\vec{x}: V_G \rightarrow \Labels$, a pairwise utility function is defined as $\sum_{u \in V_G} \utility(u, \vec{x}) + \sum_{e \in E_G} \utility(e, \vec{x})$, where for a vertex $u$, the function $\utility(u, \vec{x})$ is an arbitrary function that depends only on the label of $u$, and for each edge $e=\{u, v\}$, the function $\utility(e, \vec{x})$ is an arbitrary function that depends only on the labels of $v$ and $u$. These functions can be different for different vertices $u$ and also for different edges $e$. A pairwise cost function is defined similarly. 

For a probabilistic/fractional assignment of labels to vertices $V_G$, where vertex $v$ assumes each label in $\Sigma$ with a given probability, the utility and costs are defined as the expected values of the utility and cost functions, if we randomly draw integral labels for the vertices from their corresponding distributions (and independently, though of course each term in the summation depends only on the labels of two vertices and thus pairwise independence suffices).
\end{restatable}

\begin{restatable}{lemma}{FaouretalRounding}
  \label{lemma:rounding}
  [cf. Lemma 2.5 of \cite{faour2022local}] Let $G=(V_G,E_G)$ be a multigraph, which is equipped with utility and cost functions $\utility$ and $\cost$ and with a fractional label assignment $\lambda$ such that for every label $\alpha\in \Labels$ and every $v\in V_G$, $\lambda_\alpha(v)\geq \lambda_{\min}$ for some given value $\lambda_{\min}\in(0,1]$. Further assume that $G$ is equipped with a proper $\zeta$-vertex coloring. If $\utility(\lambda)-\cost(\lambda)\geq 0.1\utility(\lambda)$ and if each node knows the utility and cost functions of its incident edges, there is a deterministic $O\big( \log^2\big(\frac{1}{\lambda_{\min}}\big)+ \log\big(\frac{1}{\lambda_{\min}}\big)\cdot \log^* \zeta\big)$-round distributed message passing algorithm on $G$, in the \local model, that computes an integral label assignment $\ell$ for which
  \[
  \utility(\ell)-\cost(\ell) \geq 0.9 \big(\utility(\lambda) - \cost(\lambda)\big).
  \]
\end{restatable}

\subsection{The hitting set subroutines}
We also make use of a deterministic distributed algorithm for computing a certain hitting set. The formal statement is provided below. 
\begin{restatable}{lemma}{hittingSetGoodFraction}
   \label{lem:hitting_set_good_fraction}
   There exists a deterministic distributed algorithm in the \local model such that, for every $\Delta \in \mathbb{N}$, $b \in \mathbb{N}$, $k \in [1,\Delta]$, $p = \Omega(1/k)$ and $Norm \geq 0$, it provides the following guarantees:
   The input is a bipartite graph $H = (U_H \sqcup  V_H, E_H)$ with  $deg_H(u) = \Delta$ for every $u \in U_H$. Initially, each node is equipped with a unique $b$-bit identifier and each node $u$ is assigned a weight $w_u \geq 0$. Each node also knows at the beginning to which side of the bipartition it belongs. The algorithm computes a subset $V^{sub} \subseteq V_H$ satisfying
   
   \fullOnly{\[\sum_{u \in U_H \colon |N_H(u) \cap V^{sub}| \leq 0.5 \lfloor \Delta/k \rfloor} w_u + Norm \cdot |V^{sub}| \leq 4\left(e^{-pk}\sum_{u \in U_{H}} w_u + Norm \cdot p \cdot |V_H| \right).\]}
   \shortOnly{
   \begin{align*}
   \sum_{u \in U_H \colon |N_H(u) \cap V^{sub}| \leq 0.5 \lfloor \Delta/k \rfloor} w_u + Norm \cdot |V^{sub}| \\ 
   \leq 4\left(e^{-pk}\sum_{u \in U_{H}} w_u + Norm \cdot p \cdot |V_H| \right)    
   \end{align*} 
   }
      The algorithm runs in $O(k p (\log^2(k) + \log(k)\log^* b))$ rounds.
\end{restatable}

Notice that a simple probabilistic scheme where we include each node of $V_H$ in $V^{sub}$ independently with probability $p$ achieves the desired properties in expectation (and indeed with even stronger guarantees). In particular, let us even group the neighbors of each $u \in U_H$ into $\lfloor \Delta/k \rfloor$ groups each of size $k$---we will rely on independence only within each group and this simplifies the task for derandomization. The probability that a group does not have a node in $V^{sub}$ is at most $(1-p)^{k} \leq e^{-pk}$. If $u \in U_H \colon |N_H(u) \cap V^{sub}| \leq 0.5 \lfloor \Delta/k \rfloor$, then half of its groups are not hit. However, we expect only $e^{-pk}$ fraction of groups not to be hit (and the fraction can be weighted, by taking the $w_u$ weights into account). So, roughly speaking, the weighted fraction of nodes $u$ such that $u \in U_H \colon |N_H(u) \cap V^{sub}| \leq 0.5 \lfloor \Delta/k \rfloor$ is at most $2e^{-pk}$. The lemma statement has another $2$ factor of slack beyond this bound, which simplifies the argument and suffices for the applications. We present a method to deterministically compute such a set $V^{sub}$. This is by an adaptation of the method of Faour et al.~\cite{faour2022local}, and statements somewhat similar to this lemma were implicit in \cite{faour2022local, ghaffari2023netdecomp}. \fullOnly{We provide a proof for \Cref{lem:hitting_set_good_fraction} in \Cref{app:hittingSet}.}\shortOnly{The proof of \Cref{lem:hitting_set_good_fraction} appears in the full version of this paper.}

\section{Clustering}

This section is devoted to proving our two clustering results, as restated below. Please see \Cref{def:clusterings} for the related definitions. 

\clusteringall*
\clusteringconstant*

Notice the first statement clusters all nodes and takes $\tilde{O}(\alpha^2 \log n)$ rounds, while the second statement clusters only a 0.9 fraction of the vertices, in a weighted sense, but runs faster in $\tilde{O}(\alpha \log n)$ rounds. Both algorithms can be seen as a derandomization of the randomized clustering algorithm of Miller, Peng and Xu \cite{miller2013parallel}.
The high-level approach in all of the three algorithms is the same and can be found in \cref{alg:clustering}. The idea is to first compute a delay $del(v) \in \{0,1,\ldots,50\alpha\}$ for every $v \in V$. Then, these delays define a partition as follows: First, each node starts to broadcast a token at time $del(v)$. Then, in the following $O(\alpha)$ rounds, a node forwards all the token it has received in the previous rounds to its neighbors. For a given node $u$, let $c_u$ be the node whose token reaches $u$ first, breaking ties by choosing the node with the smallest identifier. Then, $u$ joins the cluster corresponding to the node $c_u$. Formally speaking, we define $c_u := \arg \min_{v \in V} (del(v) + d(v,u),ID(v))$ where pairs are compared according to the lexicographic order and for a node $v$ we define the cluster associated with $v$ as $C_v := \{u \in V \colon c_u = v\}$. Finally, the output clustering consists of all non-empty clusters $C_v$ for $v \in V$.
To compute the delay $del(v)$ for every node $v \in V$, all three algorithms compute a sequence of sets
$V^{active}_0 := V \supseteq V^{active}_1 \supseteq \ldots \supseteq V^{active}_{10\alpha} = \emptyset$.
Then, $i_v$ is defined as the largest index $i$ with $v \in V^{active}_i$ and the delay of $v$ is defined as $del(v) := 50\alpha - 5i_v$.
Our two clustering algorithms and the randomized clustering of Miller, Peng, and Xu only differ in the way $V_{i+1}^{active}$ is computed given the set $V^{active}_i$; the randomized clustering algorithm simply obtains $V_{i+1}^{active}$ by including each vertex in $V^{active}_i$ independently with probability $p_{rand} = 2^{-\log(n)/\alpha}$.
Note that this implies $V^{active}_{10\alpha} = \emptyset$ with high probability.
Furthermore, since $del(v) \in [0,50\alpha]$, the diameter of the output partition is $O(\alpha)$.

\begin{claim}
    \label{claim:clustering_diameter}
    \cref{alg:clustering} produces a partition with diameter $O(\alpha)$.
\end{claim}
\fullOnly{
\begin{proof}
Consider some arbitrary $u \in V$. We have

\[del(c_u) + d(c_u,u) \leq del(u) + d(u,u) \leq 50 \alpha\]

and therefore $d(c_u,u) \leq 50 \alpha$. Let $w$ be an arbitrary node on some shortest path from $c_u$ to $u$. We show that $c_w = c_u$, which together with $d(c_u,u) \leq 50 \alpha$ finishes the proof of the claim. Consider some arbitrary node $v$. We have to show that 

\[(del(c_u) + d(c_u,w),ID(c_u)) < del(v) + d(v,w),ID(v)).\]

As $w$ is on a shortest path from $c_u$ to $u$, we have $d(c_u,u) = d(c_u,w) + d(w,u)$. Therefore,

\begin{align*}
    (del(c_u) + d(c_u,w),ID(c_u)) &= (del(c_u) + d(c_u,u) - d(w,u),ID(c_u)) \\
    &< (del(v) + d(v,u) - d(w,u),ID(v)) \\
    &\leq (del(v) + d(v,w),ID(v)).
\end{align*}
\end{proof}
}

	\begin{restatable}{algorithm}{delays}
	\caption{Generic Partitioning Algorithm}
	\label{alg:clustering}
	Input: Parameter $\alpha \in \mathbb{N}$, Algorithm $\mathcal{A}$\\
	Output: A partition $\fC$ of the vertex set
	\begin{algorithmic}[1] 
		\Procedure{Partitioning}{} 
		\State \Comment{Delay Computation}
		\State $V_0^ {active} \leftarrow V$
		\For{$i \leftarrow 0,1, \ldots, 10\alpha - 1$}
	        \State $V^{active}_{i+1} \leftarrow \mathcal{A}(V^{active}_i, i, \alpha)$ \Comment{$V^{active}_{i+1} \subseteq V^{active}_i$}  
		\EndFor
		\State $Assert: V^{active}_{10 \alpha} = \emptyset$
		\State For every $v \in V$, let $i_v$ be the largest index $i$ with $v \in V^{active}_i$
		\State For every $v \in V$, define $del(v) \leftarrow 50\alpha - 5i_v$ 
		\State \Comment{Partitioning given the delays}
		\State $\forall u \in V \colon c_u \leftarrow \arg \min_{v \in V} (del(v) + d(v,u), ID(v))$ 
		\State $\forall v \in V \colon C_v \leftarrow \{u \in V \colon c_u = v\}$
		\State \Return $\fC = \{C_v \colon v \in V, C_v \neq \emptyset\}$
		\EndProcedure
	\end{algorithmic}
\end{restatable}

The next lemma will be used later on to upper bound the cluster degree of a vertex. Before stating it, we introduce one more definition.

\begin{definition}[$S_i(u)$]
For every $u \in V$ and $i \in [0, 10\alpha - 1]$, we define
 \[S_i(u) = V^{active}_i \cap B_G(u,d(u,\min(V^{active}_i) + 2,100\alpha)).\]
\end{definition}

Informally speaking, the set $S_i(u)$ contains those active nodes  in $V^{active}_i$ which are almost as close to $u$ as the closest active node in $V^{active}_i$, unless the closest node in $V^{active}_i$ is far away. The last part ensures that each node $u$ can compute the set $S_i(u)$ in $O(\alpha)$ rounds, which will be important later on.

\begin{lemma}
\label{lem:clustering}
Let $u \in V$ be arbitrary, $\fC$ the partition returned by \cref{alg:clustering} and

\[R_u := \max_{i \in \{0,1,\ldots,10\alpha - 1\} \colon S_i(u) \cap V^{active}_{i+1} = \emptyset} |S_i(u)|.\]
Then, it holds that $deg_\fC(u) = O(\alpha R_u)$.
\end{lemma}

\paragraph{Randomized Intuition}
\fullOnly{Before we prove \cref{lem:clustering}, let's get some intuition by applying it to the randomized clustering algorithm.}\shortOnly{The proof of \cref{lem:clustering} is deferred to the full version. 
Here, we discuss some intuition by applying the lemma to the randomized clustering algorithm.} 
In particular, a simple hitting set argument on top of \cref{lem:clustering} gives that $deg_{\fC}(u) = O(\log n /p_{rand})$, w.h.p.
Recall that the randomized algorithm obtains $V^{active}_{i+1}$ by including each vertex in $V^{active}_i$ with probability $p_{rand} = 2^{- \log(n)/\alpha}$.
Consider an arbitrary $i \in \{0,1,\ldots,10\alpha -1\}$ and assume that $|S_i(u)| \geq 10 \log(n)/p_{rand}$. With high probability, at least one node in $S_i(u)$ will be sampled and added to $V^{active}_{i+1}$.
Hence, with high probability $R_u \leq 10 \log(n)/p_{rand}$, in which case \cref{lem:clustering} gives $\deg_\fC(u) = O(\alpha \cdot \log(n)2^{log(n)/\alpha})$. We note that this analysis is not tight.
Using the memoryless property of the exponential distribution, one can show that $\deg_\fC(u) = O(2^{\log(n)/\alpha})$ with positive constant probability and $\deg_\fC(u) = O(\log (n) 2^{\log(n)/\alpha})$ with high probability.
The latter bound improves our bound by a factor of $\alpha$, i.e., for $\alpha = O(\log n)$ our proof gives that with high probability $\deg_\fC(u) = O(\log^2 n)$ while a tighter analysis shows that $\deg_{\fC}(u) = O(\log n)$. However, this simple hitting set analysis is sufficient to obtain our desired result.

\fullOnly{
\begin{proof}[Proof of \cref{lem:clustering}]
Let $w$ be an arbitrary neighbor of $u$. Note that in order to prove \cref{lem:clustering}, it suffices to show that

\[c_w \in S_u := \bigcup_{i \in \{0,1,\ldots,10\alpha - 1\} \colon S_i(u) \cap V^{active}_{i+1} = \emptyset} S_i(u),\]

as $|S_u| \leq 10\alpha R_u$. In particular, we show below that $c_w \in S_{i_{c_w}}(u)$ and $S_{i_{c_w}}(u) \cap V^{active}_{i_{c_w} + 1} = \emptyset$.

Let $v$ be an arbitrary node.
As $u$ and $w$ are neighbors and the token of $c_w$ arrives no later at $w$ than the token of $v$, we can conclude that the token of $c_w$ arrives at $u$ at most two time units after the token of $v$. More formally,

\fullOnly{
\begin{equation*}
    del(c_w) + d(c_w,u) \leq del(c_w) + d(c_w,w) + 1 \leq del(v) + d(v,w) + 1 \leq del(v) + d(v,u) + 2.
\end{equation*}
}
\shortOnly{
\begin{align*}
    del(c_w) + d(c_w,u) &\leq del(c_w) + d(c_w,w) + 1 \\ 
    &\leq del(v) + d(v,w) + 1 \leq del(v) + d(v,u) + 2.
\end{align*}
}

In particular,
\fullOnly{
\begin{equation}
    \label{eq:clustering_distance}
d(c_w,u) \leq d(v,u) +  \left( del(v) - del(c_w) \right) +  2 \leq  d(v,u) + 5(i_{c_w} - i_v) + 2,
\end{equation}
}
\shortOnly{
\begin{align}
    \label{eq:clustering_distance}
d(c_w,u) &\leq d(v,u) +  \left( del(v) - del(c_w) \right) +  2 \nonumber \\ &\leq  d(v,u) + 5(i_{c_w} - i_v) + 2,
\end{align}
}
where the second inequality follows from $del(v) := 50\alpha - 5i_v$ and $del(w) := 50 \alpha - 5i_{c_w}$.
If $v \in V^{active}_{i_{c_w}}$, then $i_v \geq i_{c_w}$ and if $v \in V^{active}_{i_{c_w} + 1}$, then $i_v > i_{c_w}$. Therefore, \cref{eq:clustering_distance} gives

\begin{enumerate}
    \item for every $v \in V^{active}_{i_{c_w}}, d(c_w,u) \leq d(v,u) + 2$ and
    \item for every $v \in V^{active}_{i_{c_w} + 1}, d(v,u) \geq d(c_w,u) + 3$.
\end{enumerate}

The first part together with $d(c_w,u) \leq del(u) + d(u,u) + 2 \leq 50\alpha + 2$ implies that  $c_w \in S_{i_{c_w}}(u) := V^{active}_{i_{c_w}} \cap B_G(u,\min(d(u,V^{active}_{i_{c_w}}) + 2,100\alpha))$.
The second part directly implies that for every $v \in V^{active}_{i_{c_w} + 1}$ it holds that $v \notin S_{i_{c_w}}(u)$ and therefore $S_{i_{c_w}}(u) \cap V^{active}_{i_{c_w} + 1} = \emptyset$.
Together with the discussion above, this finishes the proof.
\end{proof}
}

\subsection{Clustering with Constant Fraction Good}
\label{sec:clustering_constant}

This section is dedicated to proving \cref{thm:clustering_constant}.
As a reminder, the goal is to compute a partition with diameter $O(\alpha)$ such that at least 90 percent of the vertices (in a weighted sense) have a cluster degree of $2^{\tilde{O}(\log n)/\alpha}$. 

The  algorithm follows the outline of \cref{alg:clustering}. The only missing part is to specify how to compute $V^{active}_{i+1}$ given $V^{active}_i$ for every $i \in [0,10\alpha - 1]$.

In phase $i$, the goal is to compute $V^{active}_{i+1}$ in such a way that almost all nodes $u$ with $|S_i(u) \cap V^{active}_i| = 2^{\tilde{\Omega}(\log n)/\alpha}$ satisfy $S_i(u) \cap V^{active}_{i+1} \neq \emptyset$. We refer to a node $u$ that does not satisfy this condition as being bad in phase $i$. For a node $u$, if there does not exist a phase $i$ in which $u$ is bad, \cref{lem:clustering} implies that the clustering degree of $u$ is $2^{\tilde{O}(\log n)/\alpha}$.
Below, we give the formal definition of being bad in phase $i$.

\begin{definition}[Bad in phase $i$, $U^{bad}_i$]
    For every $i \in [0,10\alpha-1]$, we refer to a node as bad in phase $i$ if $|S_i(u)| \geq \lceil 1000 \log\log(N) \rceil^{\log(N)/\alpha}$ and $S_i(u) \cap V^{active}_{i+1} = \emptyset$. We denote by $U^{bad}_i$ the set consisting of all bad nodes in phase $i$.
\end{definition}

Recall, in the \local model nodes don't have access to the precise number of vertices in the graph, but are only given a polynomial upper bound, which we denote by $N$. Also note that $\lceil 1000 \log\log(N) \rceil^{\log(N)/\alpha} = 2^{\tilde{O}(\log n)/\alpha}$.

Starting with $V^{active}_i$, we compute $V^{active}_{i+1}$ in $\log(N)/\alpha$ steps. In step $j$, we start with a set $V^{active}_{i,j}$ and compute a subset $V^{active}_{i,j+1} \subseteq V^{active}_{i,j}$. Initially, we define $V^{active}_{i,0} = V^{active}_i$ and after the last step, we set $V^{active}_{i+1} = V^{active}_{i,\log(N)/\alpha}$.
Note that we can assume without loss of generality that $N$ is a sufficiently large constant and that $\alpha$ divides $\log(N)$.

Before we explain in more detail what happens in each step, we first give two more definitions.
The first definition associates a degree threshold for each step.
 
\begin{definition}[$\deg_j$]
 For every $j \in [0, \log(N)/\alpha]$, we define $\deg_j = \lceil1000 \log \log(N)\rceil^{\log(N)/\alpha - j}$.
\end{definition}

If a node $u$ is bad in phase $i$, then $|S_i(u) \cap V^{active}_i| \geq \deg_0$ and $|S_i(u) \cap V^{active}_{i+1}| < \deg_{\log(N)/\alpha}$.

Next, we define what it means for a node to be bad during a step.

\begin{definition}[Bad in step $j$ of phase $i$, $U^{bad}_{i,j}$]
For every $i \in [0,10\alpha-1]$ and $j \in [0,log(N)/\alpha - 1]$
We say that a node is bad in step $j$ of phase $i$ if $|S_i(u) \cap V_{i,j}^{active}| \geq \deg_j$ and
$|S_i(u) \cap V_{i,j+1}^{active}| < \deg_{j+1}$.
We denote by $V_{i,j}^{bad}$ the set of bad nodes during step $j$ of phase $i$.
\end{definition}

If a node $u$ is bad in phase $i$, then $|S_i(u) \cap V^{active}_{i,0}| \geq \deg_0$ and $|S_i(u) \cap V^{active}_{i,\log(N)/\alpha}| < \deg_{\log(N)/\alpha}$. Hence, there exists some $j \in [0,\log(N)/\alpha - 1]$ such that $|S_i(u) \cap V^{active}_{i,j}| \geq \deg_j$ and $|S_i(u) \cap V^{active}_{i,j+1}| < \deg_{j+1}$ and therefore  $u$ is bad in at least one step of phase $i$. 


Fix some $i \in [0,10\alpha - 1]$ and $j \in [0,\log(N)/\alpha - 1]$. We compute $V_{i,j+1}^{active}$ given $V_{i,j}^{active}$ by first defining a bipartite graph $H$, then invoking the hitting set subroutine of \cref{lem:hitting_set_good_fraction} to deduce that there exists an efficient \local algorithm on the bipartite graph which computes a subset of $V^{active}_{i,j}$ with the desired properties, and finally use the fact that we can simulate the algorithm on the bipartite graph on the original graph with an $O(\alpha)$-multiplicative overhead.

Consider the bipartite graph $H = (U_H \sqcup V_H,E_H)$ with $U_H = \{u \in V \colon |S_i(u) \cap V^{active}_{i,j}| \geq \deg_j\}, V_H = V^{active}_{i,j}$ and where we connect each node $u \in U_H$ to $\deg_j$ nodes in $S_i(u) \cap V^{active}_{i,j}$.

Using the fact that $S_i(u) \subseteq B_G(u,1000 \alpha)$ for every $u \in V$, it follows that each round in $H$ can be simulated in $O(\alpha)$ rounds in the original graph. Moreover, as $H$ contains at most two copies of each node in $G$, we can assign each node a unique $b$-bit identifier with $b = O(\log N)$.

One of the parameters of \cref{lem:hitting_set_good_fraction} is a normalization constant, which we define below.

\begin{definition}[Normalization Constant]
For every $i \in [0,10\alpha-1]$ and $j \in [0,\log(N)/\alpha]$, we define $Norm_{i,j} = \frac{2^{i \cdot \log(N)/\alpha + j}}{N^2}$.
\end{definition}

We now invoke \cref{lem:hitting_set_good_fraction} with parameters $\Delta = \deg_j$, $k = \lceil 100 \log \log(N) \rceil$ and $p = \frac{1}{16}$ to conclude that there exists a \local algorithm running in $O(kp(\log^2(k) + \log(k) \log^* b)) = \tilde{O}(\log \log n)$ rounds on $H$, and hence can be simulated in $\tilde{O}(\log \log n)\alpha$ rounds in $G$, which computes a set $V^{sub} \subseteq V_H$ satisfying

\fullOnly{
\[\sum_{u \in U_H \colon |N_H(u) \cap V^{sub}| \leq 0.5 \lfloor \Delta/k\rfloor} x_u + Norm_{i,j} \cdot |V^{sub}| \leq 4\left(e^{-pk}\sum_{u \in U_H} x_u + Norm_{i,j} \cdot p \cdot |V_H| \right).\]
}\shortOnly{
\begin{align*}
\sum_{u \in U_H \colon |N_H(u) \cap V^{sub}| \leq 0.5 \lfloor \Delta/k\rfloor} x_u + Norm_{i,j} \cdot |V^{sub}| \\
\leq 4\left(e^{-pk}\sum_{u \in U_H} x_u + Norm_{i,j} \cdot p \cdot |V_H| \right).    
\end{align*}
}
We then define $V^{active}_{i,j} = V^{sub}$. Consider some node $u$ that is bad in step $j$ of phase $i$. By definition, $|S_i(u) \cap V^{active}_{i,j}| \geq \deg_j$ and $|S_i(u) \cap V^{active}_{i,j+1}| < \deg_{j+1}$. Therefore, $u \in U_H$ and 

\fullOnly{
\[|N_H(u) \cap V^{sub}| \leq |S_i(u) \cap V^{active}_{i,j+1}| < \deg_{j+1} = \frac{\deg_j}{\lceil 1000 \log\log(N) \rceil} \leq 0.5 \lfloor \Delta/k \rfloor.\]
}
\shortOnly{
\begin{align*}
|N_H(u) \cap V^{sub}| &\leq |S_i(u) \cap V^{active}_{i,j+1}| < \deg_{j+1} \\ &= \frac{\deg_j}{\lceil 1000 \log\log(N) \rceil} \leq 0.5 \lfloor \Delta/k \rfloor    
\end{align*}
}

Since $e^{-pk} = e^{-(100/16)\log \log N} \leq \frac{1}{800 \log(N)}$, we can conclude that

\fullOnly{
\[\sum_{u \in U^{bad}_{i,j}} x_u + Norm_{i,j} \cdot |V^{active}_{i,j+1}| \leq \frac{1}{200 \log(N)} \sum_{u \in V } x_u + \frac{Norm_{i,j}}{4}|V^{active}_{i,j}|.\]
}
\shortOnly{
\begin{align*}
& &&\sum_{u \in U^{bad}_{i,j}} x_u + Norm_{i,j} \cdot |V^{active}_{i,j+1}| \\ \leq & &&\frac{1}{200 \log(N)} \sum_{u \in V } x_u + \frac{Norm_{i,j}}{4}|V^{active}_{i,j}|.
\end{align*}
}
From this, we can deduce the following two claims:
\begin{claim}[No active nodes in the end]
\label{claim:no_active_nodes}
For every $i \in [0,10\alpha - 1]$ and $j \in [0,\log(N)/\alpha]$, we have

\[|V^{active}_{i,j}| \leq \frac{\sum_{u \in V} x_u}{50\log(N) Norm_{i,j}},\]
which in particular also implies that $V^{active}_{10 \alpha} = \emptyset$.
\end{claim}
\begin{proof}

First, note that $$|V^{active}_{0,0}| \leq n \leq \frac{1}{50\log(N)(1/N^2)} \leq \frac{\sum_{u \in V} x_u}{50 \log(N) Norm_{0,0}}.$$ Also, recall that for every $i \in [1,10\alpha - 1]$, $V^{active}_{i+1,0} := V^{active}_{i+1} := V^{active}_{i,\log(N)/\alpha}$. Together with $Norm_{i+1,0} = Norm_{i,\log(N)/\alpha}$, we get that $|V^{active}_{i,\log(N)/\alpha}| \leq \frac{\sum_{u \in V} x_u}{50\log(N) Norm_{i,\log(N)/\alpha}}$ trivially implies that $|V^{active}_{i+1,0}| \leq \frac{\sum_{u \in V} x_u}{50\log(N) Norm_{i+1,0}}$.

Now, assume that $|V^{active}_{i,j}| \leq \frac{\sum_{u \in V} x_u}{50\log(N) Norm_{i,j}}$ for a given $i \in [0,10\alpha - 1]$ and $j \in [0,\log(N)/\alpha - 1]$. This implies that

\fullOnly{
\[|V^{active}_{i,j+1}| \leq \frac{\sum_{u \in V} x_u}{200 \log(N) Norm_{i,j}} + \frac{1}{4}|V^{active}_{i,j}| \leq \frac{\sum_{u \in V} x_u}{100 \log(N) Norm_{i,j}}= \frac{\sum_{u \in V} x_u}{50 \log(N) Norm_{i,j+1}}\sum_{u \in V} x_u.\]
}
\shortOnly{
\begin{align*}
|V^{active}_{i,j+1}| &\leq \frac{\sum_{u \in V} x_u}{200 \log(N) Norm_{i,j}} + \frac{1}{4}|V^{active}_{i,j}| \\ &\leq \frac{\sum_{u \in V} x_u}{100 \log(N) Norm_{i,j}} = \frac{\sum_{u \in V} x_u}{50 \log(N) Norm_{i,j+1}}\sum_{u \in V} x_u.    
\end{align*}
}

Thus, $|V^{active}_{i,j}| \leq \frac{\sum_{u \in V} x_u}{50\log(N) Norm_{i,j}}$, as one can formalize by an induction.
In particular,

\[|V^{active}_{10 \alpha - 1,\log(N)/\alpha}| \leq \frac{\sum_{u \in V} x_u}{50 \log(N) Norm_{10 \alpha - 1,\log(N)/\alpha}} < 1\]
\end{proof}

and as $V^{active}_{10 \alpha} := V^{active}_{10 \alpha - 1,\log(N)/\alpha}$, we get that $V^{active}_{10 \alpha} = \emptyset$.

\begin{claim}[Total number of bad nodes in step $j$ of phase $i$]
\label{claim:bad_nodes}
For every $i \in [0,10\alpha - 1]$ and $j \in [0,\log(N)/\alpha - 1]$, we have

\[\sum_{u \in U^{bad}_{i,j}} x_u \leq \frac{1}{100 \log(N)} \sum_{u \in V} x_u.\]
\end{claim}

\begin{proof}
We have

\[\sum_{u \in U^{bad}_{i,j}} x_u \leq \frac{1}{200 \log(N)} \sum_{u \in V } x_u + \frac{Norm_{i,j}}{4}|V^{active}_{i,j}| \leq \frac{\sum_{u \in V} x_u}{100 \log (N)}.\]
\end{proof}

\begin{proof}[Proof of \cref{thm:clustering_constant}]

According to \cref{claim:clustering_diameter}, the diameter of the output partition is $O(\alpha)$.

\cref{lem:clustering} implies that $V \setminus (\cup_{i=0}^{10\alpha - 1} U^{bad}_i) \subseteq \{u \in V \colon \deg_\fC(u) = 2^{\tilde{O}(\log n)/\alpha}\}$ and \cref{claim:bad_nodes} implies that
$\sum_{u \in \cup_{i=0}^{10\alpha - 1}\cup_{j=0}^{\log(N)/\alpha} U^{bad}_{i,j}} x_u \leq 10 \alpha \cdot (\log(N)/\alpha)\frac{\sum_{u \in V} x_u}{100 \log(N)} = \frac{\sum_{u \in V} x_u}{10}$.
Together with $U^{bad}_i \subseteq \cup_{j=0}^{\log(N)/\alpha} U^{bad}_{i,j}$ for every $j \in [0,\log(N)/\alpha - 1]$, we therefore get $\sum_{u \in V \colon \deg_{\fC}(u) = 2^{\tilde{O}(\log n)/\alpha}} x_u \geq 0.9 \sum_{u \in V} x_u $, as needed.

We have argued that each step takes $\tilde{O}(\log \log n \cdot \alpha)$ rounds. There are $O(\alpha) \cdot O(\log n /\alpha )$ steps in total, and thus the total runtime to compute the delays is $\tilde{O}(\alpha \log n)$. Given the delays, the output partition can be computed in $O(\alpha)$ additional rounds. Hence, the algorithm runs in $\tilde{O}(\alpha \log n)$ rounds.

\end{proof}

\subsection{Clustering with all Nodes Good}
\label{sec:clustering_all}

This section is dedicated to proving \cref{thm:clustering_all}.
As a reminder, the goal is to compute a partition with diameter $O(\alpha)$ such that \emph{all} vertices have a cluster degree of $2^{\tilde{O}(\log n)/\alpha}$. 

As in the previous section, the algorithm follows the outline of \cref{alg:clustering}. The only missing part is to specify how to compute $V^{active}_{i+1}$ given $V^{active}_i$ for every $i \in [0,10\alpha - 1]$.

In phase $i$, the goal is to compute $V^{active}_{i+1}$ in such a way that all nodes $u$ with $|S_i(u) \cap V^{active}_i| = 2^{\tilde{\Omega}(\log n)/\alpha}$, which we refer to as being important in phase $i$, satisfy $S_i(u) \cap V^{active}_{i+1} \neq \emptyset$. By using \cref{lem:clustering}, this is sufficient to ensure that all vertices have a cluster degree of $2^{\tilde{O}(\log n)/\alpha}$. We start with the formal definition of being important in phase $i$. Recall that we use $N$ to denote a polynomial upper bound on the number of nodes, which is known to all the nodes.

\begin{definition}[Important in phase $i$, $U^{important}_i$]
    For every $i \in [0,10\alpha - 1]$, we refer to a node as being important in phase $i$ if $|S_i(u)| \geq (2000\log N)^{\log(N)/\alpha}$.
\end{definition}

As before, phase $i$ consists of $\log(N)/\alpha$ steps and we compute a sequence of sets $V^{active}_{i,0} := V^{active}_i \supseteq V^{active}_{i,1} \supseteq V^{active}_{i,2} \supseteq \ldots \supseteq V^{active}_{i,\log(N)/\alpha} =: V^{active}_{i+1}$.

Compared to \cref{sec:clustering_constant}, one crucial difference is that the algorithm starts computing $V^{active}_{i,j+1}$ before it has finished computing $V^{active}_{i,j}$. 

As before, we associate a degree threshold with each step. We remark that the degree threshold is slightly larger compared to the one in \cref{sec:clustering_constant}.

\begin{definition}[$\deg_j$]
	For every $j \in [0,\log(N)/\alpha]$, we define $\deg_j = (2000\log N)^{\log(N)/\alpha - j}$.
\end{definition}

Note that $u$ is important in phase $i$ if $|S_i(u)| \geq \deg_0$ and that $S_i(u) \cap V^{active}_{i+1} \neq \emptyset$ is equivalent to $|S_i(u) \cap V^{active}_{i+1}| = |S_i(u) \cap V^{active}_{i,\log(N)/\alpha}|\geq \deg_{\log(N)/\alpha}$. In general, each node $u$ which is important in phase $i$ will satisfy $|S_i(u) \cap V^{active}_{i,j}| \geq \deg_j$ for every $j \in [0,\log(N)/\alpha]$.

For each $i \in [0,10\alpha - 1]$ and $j \in [1,\log(N)/\alpha]$, we compute $V^{active}_{i,j}$ by computing a sequence of sets $V^{active}_{i,j,0} := \emptyset \subseteq V^{active}_{i,j,1} \subseteq \ldots \subseteq V^{active}_{i,j,4\log(N)} =: V^{active}_{i,j}$. To simplify notation, we also define 
$V^{active}_{i,0,0} = V^{active}_{i,0,1} = \ldots = V^{active}_{i,0,4\log(N)} = V^{active}_{i,0} := V^{active}_i$ for every $i \in [0,10\alpha - 1]$.

For every $j \in [0,\log(N)/\alpha]$ and  $\ell \in [0,4\log(N)]$, we define a partition $U^{important}_i =  U^{good}_{i,j,\ell} \sqcup U^{bad}_{i,j,\ell}$.

\begin{definition}[$U^{good}_{i,j,\ell}$, $U^{bad}_{i,j,\ell}$]
    For every $i \in [0,10\alpha - 1], j \in [0,\log(N)/\alpha], \ell \in [0,4\log(N)]$, we define $$U^{good}_{i,j,\ell} = \{u \in U^{important}_i \colon |S_i(u) \cap V^{active}_{i,j,\ell}| \geq \deg_j\}$$  and $$U^{bad}_{i,j,\ell} = U^{important}_i \setminus U^{good}_{i,j,\ell}.$$
\end{definition}


\paragraph{Computing $V^{active}_{i,j,\ell}$ given $V^{active}_{i,j - 1,\ell}$ and $V^{active}_{i,j,\ell - 1}$}

Now, fix some $i \in [0,10\alpha - 1]$, $j \in [1,\log(N)/\alpha]$ and $\ell \in [1,4\log(N)]$.
We compute $V^{active}_{i,j,\ell}$ given $V^{active}_{i,j-1,\ell}$ and $V^{active}_{i,j,\ell-1}$ by first defining a bipartite graph $H$, then invoking the hitting set subroutine of \cref{lem:hitting_set_good_fraction} to deduce that there exists an efficient \local algorithm on the bipartite graph which computes a subset of $V^{active}_{i,j-1,\ell}$ with the desired properties, and finally use the fact that we can simulate the algorithm on the bipartite graph on the original graph with an $O(\alpha)$-multiplicative overhead.
Consider the bipartite graph $H = (U_H \sqcup V_H, E_H)$ with $U_H = U^{good}_{i,j-1,\ell} \cap U^{bad}_{i,j,\ell-1}$, $V_H = V^{active}_{i,j-1,\ell}$ and where we connect each node $u \in U_H$ to $\deg_{j-1}$ nodes in $S_i(u) \cap V^{active}_{i,j-1,\ell}$. Similar as in \cref{sec:clustering_constant}, each round in $H$ can be simulated in $O(\alpha)$ rounds in the original graph and we can assume that each node is assigned a unique $b$-bit identifier with $b = O(\log N)$.
We also define $Norm_{i,j} = 2^{i \cdot \log(N)/\alpha + j}$ for every $i \in [0,10\alpha-1]$ and $j \in [0,\log(N)/\alpha]$.

We now invoke \cref{lem:hitting_set_good_fraction} with parameters $\Delta = \deg_{j-1}$, $k = 500 \log(N)$, $p = \frac{1}{64 \log N}$ and $Norm = Norm_{i,j} 2^{j - \ell}$ to conclude that there exists a \local algorithm running in $O(\log^2 \log n)$ rounds on $H$, and hence can be simulated in $O(\alpha \log^2 \log n)$ rounds in $G$, which computes a set $V^{sub}_{i,j,\ell} \subseteq V_H$ satisfying
\fullOnly{
\[|\{u \in U_H \colon |N_H(u) \cap V^{sub}_{i,j,\ell} | \leq 0.5 \lfloor \Delta/k\rfloor\}| + 2^{j-\ell} Norm_{i,j} \cdot |V^{sub}_{i,j,\ell}| \leq 4\left(e^{-pk} |U_H| + 2^{j - \ell}Norm_{i,j} \cdot p \cdot |V_H| \right).\]
}
\shortOnly{
\begin{align*}
|\{u \in U_H \colon |N_H(u) \cap V^{sub}_{i,j,\ell} | \leq 0.5 \lfloor \Delta/k\rfloor\}| + 2^{j-\ell} Norm_{i,j} \cdot |V^{sub}_{i,j,\ell}| \\ \leq 4\left(e^{-pk} |U_H| + 2^{j - \ell}Norm_{i,j} \cdot p \cdot |V_H| \right).    
\end{align*}
}

We then define $V^{active}_{i,j,\ell} = V^{active}_{i,j,\ell - 1} \cup V^{sub}_{i,j,\ell}$.

Consider some $u \in U^{good}_{i,j-1,\ell} \cap U^{bad}_{i,j,\ell}$. As $u \in U^{bad}_{i,j,\ell}$ implies $u \in U^{bad}_{i,j,\ell - 1}$, we have $u \in U^{good}_{i,j-1,\ell} \cap U^{bad}_{i,j,\ell - 1} =: U_H$. Moreover,
\fullOnly{
\[|N_H(u) \cap V^{sub}_{i,j,\ell}| \leq |S_i(u) \cap V^{active}_{i,j,\ell}| < \deg_j = \frac{\deg_{j-1}}{2000 \log(N)} \leq 0.5 \lfloor \frac{\Delta}{k}\rfloor.\]
}
\shortOnly{
\begin{align*}
  |N_H(u) \cap V^{sub}_{i,j,\ell}| &\leq |S_i(u) \cap V^{active}_{i,j,\ell}| < \deg_j \\ 
  &= \frac{\deg_{j-1}}{2000 \log(N)} \leq 0.5 \lfloor \frac{\Delta}{k}\rfloor.  
\end{align*}
}

Hence, using that $e^{-pk} \leq \frac{1}{32}$, we get

\fullOnly{
\[2^{\ell - j}|U^{good}_{i,j-1,\ell} \cap U^{bad}_{i,j,\ell}| +  Norm_{i,j} \cdot |V^{sub}_{i,j,\ell}| \leq \frac{2^{\ell - j}}{8}|U^{bad}_{i,j,\ell-1}| + \frac{Norm_{i,j}}{16 \log(N)} |V^{active}_{i,j-1,\ell}|.\]
}\shortOnly{
\begin{align*}
   2^{\ell - j}|U^{good}_{i,j-1,\ell} \cap U^{bad}_{i,j,\ell}| +  Norm_{i,j} \cdot |V^{sub}_{i,j,\ell}| \\ \leq \frac{2^{\ell - j}}{8}|U^{bad}_{i,j,\ell-1}| + \frac{Norm_{i,j}}{16 \log(N)} |V^{active}_{i,j-1,\ell}|. 
\end{align*}
}

The claim below is the key claim in the analysis.

\begin{claim}
\label{claim:clustering_all_key_claim}
For every $i \in [0,10\alpha-1], j \in [0,\log(N)/\alpha]$ and $\ell \in [0,4\log(N)]$, it holds that
\begin{enumerate}
    \item $|V^{active}_{i,j}| \leq \frac{2 \log(N) \cdot n}{Norm_{i,j}}$ and
    \item $|U^{bad}_{i,j,\ell}| \leq n \cdot 2^{j - \ell}$.
\end{enumerate}
In particular, $V^{active}_{10\alpha} = \emptyset$ and $\{u \in U^{important}_i \colon S_i(u) \cap V^{active}_{i+1} = \emptyset\} = \emptyset$ for every $i \in [0,10\alpha - 1]$.
\end{claim}

\begin{proof}
	
First, consider some fixed $i \in [0,10\alpha - 1]$ and $j \in [1,\log(N)/\alpha]$ and assume $|V^{active}_{i,j-1}| \leq \frac{2 \log(N) \cdot n}{Norm_{i,j-1}}$ and $|U^{bad}_{i,j-1,\ell}| \leq n \cdot 2^{(j-1) - \ell}$ for every $\ell \in [0,4\log(N)]$. In particular, for every $\ell \in [0,4\log(N) - 1]$, $V^{active}_{i,j-1,\ell} \subseteq V^{active}_{i,j-1}$ directly gives

\[\frac{Norm_{i,j}}{16 \log(N)}|V^{active}_{i,j-1,\ell}| \leq \frac{Norm_{i,j}}{16 \log(N)}\frac{2 \log(N) \cdot n}{Norm_{i,j-1}} = \frac{n}{4}.\]

We first show by induction on $\ell$ that this implies $|U^{bad}_{i,j,\ell}| \leq n \cdot 2^{j - \ell}$ for every $\ell \in [0,4\log(N)]$ and afterwards we show that $|V^{active}_{i,j}| \leq \frac{2\log(N) \cdot n}{Norm_{i,j}}$.

The base case $\ell = 0$ trivially holds as $|U^{bad}_{i,j,0}| \leq n \leq n \cdot 2^{j - 0}$.
Now, fix some $\ell \in [1,4\log(N)]$ and assume that $|U^{bad}_{i,j,\ell-1}| \leq 2^{j-(\ell - 1)}$.
First, note that $U^{bad}_{i,j,\ell} \subseteq V^{important}_i = U^{good}_{i,j-1,\ell} \sqcup U^{bad}_{i,j-1,\ell}$ together with the initial assumption $|U^{bad}_{i,j-1,\ell}| \leq 2^{(j-1)-\ell}$ implies
\[|U^{bad}_{i,j,\ell}| \leq |U^{bad}_{i,j-1,\ell}| + |U^{good}_{i,j-1,\ell} \cap U^{bad}_{i,j,\ell}| \leq  \frac{n \cdot 2^{j-\ell}}{2} + |U^{good}_{i,j-1,\ell} \cap U^{bad}_{i,j,\ell}|.\]

Moreover, using induction, we get
\fullOnly{
\[|U^{good}_{i,j-1,\ell} \cap U^{bad}_{i,j,\ell}| \leq \frac{1}{8}|U^{bad}_{i,j,\ell-1}| + 2^{j-\ell}\frac{Norm_{i,j}}{16 \log(N)} |V^{active}_{i,j-1,\ell}| \leq \frac{n2^{j - (\ell-1)}}{8} + \frac{n \cdot 2^{j-\ell}}{4} \leq \frac{n \cdot 2^{j - \ell}}{2}\]
}\shortOnly{
\begin{align*}
   |U^{good}_{i,j-1,\ell} \cap U^{bad}_{i,j,\ell}| &\leq \frac{1}{8}|U^{bad}_{i,j,\ell-1}| + 2^{j-\ell}\frac{Norm_{i,j}}{16 \log(N)} |V^{active}_{i,j-1,\ell}| \\ &\leq \frac{n2^{j - (\ell-1)}}{8} + \frac{n \cdot 2^{j-\ell}}{4} \leq \frac{n \cdot 2^{j - \ell}}{2} 
\end{align*}
}

and therefore $|U^{bad}_{i,j,\ell}| \leq n \cdot 2^{j-\ell}$, finishing the induction. Now, consider some $\ell \in [1,4 \log(N)]$. Using $|U^{bad}_{i,j,\ell-1}| \leq n \cdot 2^{j - (\ell - 1)}$, we get
\fullOnly{
\[|V^{sub}_{i,j,\ell}| \leq \frac{1}{Norm_{i,j}} \left( \frac{2^{\ell - j}}{8}|U^{bad}_{i,j,\ell-1}| + \frac{Norm_{i,j}}{16 \log(N)} |V^{active}_{i,j-1,\ell}| \right) \leq \frac{(n/4) + (n/4)}{Norm_{i,j}} = \frac{n}{2Norm_{i,j}}.\]
}\shortOnly{
\begin{align*}
|V^{sub}_{i,j,\ell}| &\leq \frac{1}{Norm_{i,j}} \left( \frac{2^{\ell - j}}{8}|U^{bad}_{i,j,\ell-1}| + \frac{Norm_{i,j}}{16 \log(N)} |V^{active}_{i,j-1,\ell}| \right) \\ &\leq \frac{(n/4) + (n/4)}{Norm_{i,j}} = \frac{n}{2Norm_{i,j}}.
\end{align*}
}
As $V^{active}_{i,j} = \bigcup_{\ell = 1}^{4 \log(N)} V^{sub}_{i,j,\ell}$, we conclude $|V^{active}_{i,j}| \leq \frac{2 \log(N) \cdot n}{Norm_{i,j}}$.

Consider some fixed $i \in [0,10\alpha - 1]$ and assume that $|V^{active}_{i,0}| \leq \frac{2 \log(N) \cdot n}{Norm_{i,0}}$. We use induction on $j$ to conclude that $|V^{active}_{i,j}| \leq \frac{2\log(N) \cdot n}{Norm_{i,j}}$ and $|U^{bad}_{i,j,\ell}| \leq n \cdot 2^{j-\ell}$  for every $j \in [0,\log(N)/\alpha]$ and $\ell \in [0,4\log(N)]$.
For the base case $j = 0$, note that $|V^{active}_{i,0}| \leq \frac{2\log(N) \cdot n}{Norm_{i,0}}$ is just what we assumed and as $V^{active}_{i,0,\ell} := V^{active}_i$ for every $\ell \in [0,4\log(N)]$, we have $|U^{bad}_{i,0,\ell}| = 0 \leq n \cdot 2^{0 - \ell}$.
We already did the induction step going from $j-1$ to $j$ above and therefore we are done.

Finally, we prove \cref{claim:clustering_all_key_claim} by induction on $i$. For the base case $i = 0$, note that $|V^{active}_{0,0}| \leq n \leq \frac{2\log(N) \cdot n}{Norm_{0,0}}$. By using the previous induction on $j$, this is enough to show that the bounds of \cref{claim:clustering_all_key_claim} hold for $i = 0$. Now, consider some $i \in [1,10\alpha - 1]$ and assume that the bounds hold for $i-1$.
In particular, 

\[|V^{active}_{i,0}| = |V^{active}_{i-1,\log(N)/\alpha}| \leq \frac{2\log(N) \cdot n}{Norm_{i-1,\log(N)/\alpha}} = \frac{2\log(N) \cdot n}{Norm_{i,0}}.\]

Now, we can again use the previous induction on $j$ to finish the induction step.

To conclude \cref{claim:clustering_all_key_claim}, note that $|V^{active}_{10\alpha - 1,\log(N)/\alpha}| \leq \frac{2\log(N) \cdot n}{Norm_{0,0}} < 1$ and as $V^{active}_{10\alpha} := V^{active}_{10\alpha - 1,\log(N)/\alpha}$, we get that $V^{active}_{10 \alpha} = \emptyset$.
Consider some $i \in [0,10\alpha - 1]$. Note that $|U^{bad}_{i,\log(N)/\alpha,4\log(N)}| \leq n \cdot 2^{\log(N)/\alpha - 4\log(N)} < 1$ and therefore $U^{bad}_{i,\log(N)/\alpha,4\log(N)} = \emptyset$. In particular, there does not exist a node $u \in U^{important}_i$ with $|S_i(u) \cap V^{active}_{i,\log(N)/\alpha,4\log(N)}|< \deg_{\log(N)/\alpha} = 1$. As $V^{active}_{i+1} := V^{active}_{i,\log(N)/\alpha} := V^{active}_{i,\log(N)/\alpha,4\log(N)}$, we therefore get that there exists no node $u \in U^{important}_i$ with $|S_i(u) \cap V^{active}_{i+1} = \emptyset$, as desired.
\end{proof}

\begin{proof}[Proof of \cref{thm:clustering_all}]
According to \cref{claim:clustering_diameter}, the diameter of the output partition is $O(\alpha)$. Moreover, combining \cref{claim:clustering_all_key_claim} with \cref{lem:clustering} shows that the clustering of each node is $2^{\tilde{O}(\log n)/\alpha}$.
It remains to argue about the round complexity of the algorithm. Fix some $i \in [0,10 \alpha -1]$. Note that given $V^{active}_i$, we can compute $V^{active}_{i,j,0}$ for every $j \in [0,\log(N)/\alpha]$ and $V^{active}_{i,0,\ell}$ for every $\ell \in [1,4\log(N)]$ without further communication. Moreover, we have seen above that given $V^{active}_{i,j-1,\ell}$ and $V^{active}_{i,j,\ell-1}$, we can compute $V^{active}_{i,j,\ell}$ in $O(\alpha \log^2 \log n)$ rounds, for every $j \in [1,\log(N)/\alpha]$ and $\ell \in [1,4\log(N)]$.
Hence, a simple induction shows that given $V^{active}_i$, we can compute $V^{active}_{i,j,\ell}$ in $(j + \ell)O(\alpha \log^2 \log n)$ rounds for every $j \in [0,\log(N)/\alpha]$ and $\ell \in [0,4\log(N)]$ and therefore we can compute $V^{active_{i+1}}$ in $(\log(N)/\alpha + 4\log(N)) \cdot O(\alpha \log^2\log n) = \tilde{O}(\alpha \log n)$ rounds. Hence, the overall algorithm runs in $\tilde{O}(\alpha^2 \log n)$ rounds, as desired.
\end{proof}

\fullOnly{
\section{Matching and MIS}
\label{sec:MMandMIS}
In this section, we use the low-diameter clustering results obtained in the previous section in order to obtain faster deterministic algorithms for MIS and $(1+\eps)$-approximate maximum matching.
}
\shortOnly{
\section{MIS}
\label{sec:MMandMIS}
In this section, we use the low-diameter clustering results obtained in the previous section to obtain a deterministic distributed algorithm that computes an MIS in $\tilde{O}(\log^2 n)$ rounds of the \LOCAL model. 
}

\fullOnly{
\subsection{Approximate Matching}

We start with the matching algorithm. In particular, this section is devoted to proving the theorem below.

\approximatematching*

\paragraph{Notation}

Let $E$ be an arbitrary set of edges and $v$ be an arbitrary node. We define $E(v) = \{e \in E \colon v \in e\}$ as the set consisting of those edges in $E$ that are incident to $v$. As before, $N$ is a polynomial upper bound on the number of vertices given to the algorithm.

\paragraph{Computing a Constant Approximate Fractional Matching}
A fractional matching assigns each edge $e \in E$ a value $x_e \in \mathbb{R}_{\geq 0}$ such that for every $v \in V$, $\sum_{e \in E(v)} x_e \leq 1$. 
It is well-known how to compute a fractional matching in $O(\log n)$ rounds such that $\sum_{e \in E} x_e \geq \frac{1}{5}|M^*|$, where $M^* \subseteq E$ is a maximum matching of the input graph $G = (V,E)$. Moreover, one can also ensure that $x_e \geq \frac{1}{n}$ for every edge $e \in E$, which is a technicality needed later. See for example \cite{fischer2020improved}.

\paragraph{Computing a Low-Diameter Clustering}
For every $v \in V$, let $y_v := \sum_{e \in E(v)} x_e$. 
Using \cref{thm:clustering_constant} with $\alpha = \log^{1/3}(N)$, we can conclude that there exists some function $f(n) = 2^{\tilde{O}(\log^{2/3}(n))}$ and a \local algorithm running in $\tilde{O}(\log^{4/3}n)$ rounds which computes a partition $\fC$ with diameter $O(\log^{1/3} n)$ such that a large constant fraction of vertices, weighted by the $y_v$'s, has a cluster degree of at most $f(n)$. More precisely, let $V^{good} := \{v \in V \colon \deg_\fC(v) \leq f(n)\}$ be the set of nodes with cluster degree at most $f(n)$, then it holds that $\sum_{v \in V^{good}} y_v \geq 0.9 \sum_{v \in V} y_v$. Note that $y_v \geq \frac{1}{n}$ for every $v \in V$, which is one of the assumptions of \cref{thm:clustering_constant}. This directly follows from $x_e \geq \frac{1}{n}$ for every $e \in E$ and that we can assume without loss of generality that the input graph does not have any isolated vertices.

Now, let $E^{good}$ be defined as the set of edges in $E$ with both endpoints being contained in $V^{good}$. The following simple claim implies that we can restrict our attention to edges where both endpoints have a small clustering degree.

\begin{claim}
    It holds that $\sum_{e \in E^{good}} x_e \geq 0.8\sum_{e \in E}x_e$.
\end{claim}
\begin{proof}
	We have
 \fullOnly{
	\begin{align*}
		\sum_{e \in E^{good}} x_e \geq \sum_{e \in E} x_e - \sum_{v \in V \setminus V^{good}} y_v \geq \sum_{e \in E} x_e - 0.1 \sum_{v \in V} y_v &= \sum_{e \in E} x_e - 0.2 \sum_{e \in E} x_e = 0.8\sum_{e \in E} x_e. 
	\end{align*}
 }
 \shortOnly{
 \begin{align*}
 \sum_{e \in E^{good}} x_e &\geq \sum_{e \in E} x_e - \sum_{v \in V \setminus V^{good}} y_v \geq \sum_{e \in E} x_e - 0.1 \sum_{v \in V} y_v \\
 &= \sum_{e \in E} x_e - 0.2 \sum_{e \in E} x_e = 0.8\sum_{e \in E} x_e. 
\end{align*}
 }
\end{proof}

\paragraph{Intra-Cluster Rounding}

In the intra-cluster step, we start with the fractional matching $x$ and turn it into a $\frac{1}{50000f(n)\log(n)}$-integral matching $x^{intra}$ supported on $E^{good}$. Moreover, the total weight of the new fractional matching $x^{intra}$ is at most a constant factor smaller compared to the weight of the fractional matching $x$.

For every cluster $C \in \fC$, let $E_C^{good}$ denote the set consisting of those edges in $E^{good}$ whose endpoint with the larger identifier is contained in $C$. 
A key property of the fractional matching $x^{intra}$ computed in the intra-cluster step is the following: for every edge $e \in E_C^{good}$, the fractional value $x^{intra}_e$ assigned to $e$ is merely a function of the fractional values assigned to the edges in $E^{good}_C$ by the fractional matching $x$. As the \local model does not restrict message sizes and the diameter of each cluster $C$ is $O(\log^{1/3} n)$, this property allows us to compute the fractional matching $x^{intra}$ in $O(\log^{1/3} n)$ rounds.

Fix some cluster $C \in \fC$. We use the probabilistic method to argue that we can assign each edge $e \in E^{good}_C$ a value $x^{intra}_e$ such that certain conditions are satisfied. These conditions only depend on the fractional values assigned to the edges in $E^{good}_C$ by the fractional matching $x$.

To that end, we introduce one random variable $X^{intra}_e$ for every edge $e \in E^{good}_C$. If $x_e \geq \frac{1}{10000f(n)\log(n)}$, we simply set $X^{intra}_e = \frac{x_e}{5}$. Otherwise, if $x_e < \frac{1}{10000f(n)\log(n)}$, we set $X^{intra}_e = \frac{1}{50000f(n) \log(n)}$ with probability $10000f(n)\log(n)x_e$ and with the remaining probability we set $X^{intra}_e = 0$, fully independently.
Note that $\E[X^{intra}_e] = \frac{x_e}{5}$ for every edge $e \in E^{good}_C$.
The following claim follows by a simple Chernoff bound.

\begin{claim}
    \label{claim:matching_chernoff}
    For any subset $E' \subseteq E^{good}_C$, we have

    \[Pr \left[ \sum_{e \in E'} X^{intra}_e \leq  \frac{1}{2}\sum_{e \in E'} x_e + \frac{1}{1000f(n)} \right] \geq 1 - \frac{1}{n^{3}} \text{ and}\]
    \[Pr \left[\sum_{e \in E'} X^{intra}_e \geq \frac{1}{10}\sum_{e \in E'} x_e - \frac{1}{1000f(n)}\right] \geq 1 - \frac{1}{n^{3}}.\]
\end{claim}

\begin{proof}
    Let $E'^{small}$ be the set of all edges $e \in E'$ such that $x_e < \frac{1}{10000f(n)\log(n)}$ and let $X := \sum_{e \in E'^{small}} X^{intra}_e$. Note that $X$ is the sum of independent random variables taking values between $0$ and $b:= \frac{1}{50000f(n)\log(n)}$ and $\E[X] = \sum_{e \in E'^{small}} \frac{x_e}{5}$. Let $t := \frac{1}{10}\sum_{e \in E'^{small}} x_e +\frac{1}{1000f(n)}$. We have $t \geq 0.5\E[X]$ and therefore we can use  the Chernoff bound variant given in \cref{cor:chernoff} to deduce
    
    \[Pr \left[ |X-\E[X]| \geq t \right] \leq 2e^{-\frac{t}{6b}} \leq n^{-3}.\]
    Now, assume that $|X-\E[X]| \leq t$.
    Then, 
    \[\sum_{e \in E'}X^{intra}_e = \sum_{e \in E' \setminus E'^{small}} X^{intra}_e + X \leq \sum_{e \in E' \setminus E'^{small}} \frac{x_e}{5} + \E[X] + t = \sum_{e \in E'} \frac{x_e}{5}  + t \leq \sum_{e \in E'} \frac{x_e}{2} + \frac{1}{1000f(n)}.\]

    Similarly,

     \[\sum_{e \in E'}X^{intra}_e \geq \sum_{e \in E' \setminus E'^{small}} X^{intra}_e + \E[X] - t = \sum_{e \in E'} \frac{x_e}{5}  - t \geq \sum_{e \in E'}\frac{x_e}{10} - \frac{1}{1000f(n)}.\]

\end{proof}

\cref{claim:matching_chernoff} together with a simple union bound implies that the following holds with strictly positive probability: for every vertex $v \in V$,

\[\frac{1}{10}\sum_{e \in E^{good}_C(v)} x_e - \frac{1}{1000f(n)} \leq \sum_{e \in E^{good}_C(v)} X^{intra}_e \leq \frac{1}{2}\sum_{e \in E^{good}_C(v)} x_e + \frac{1}{1000f(n)}.\]

Therefore, using brute-force computation, each cluster $C$ can compute in $O(\log^{1/3}(n))$ rounds a value $x^{intra}_e \in \{0\} \cup \left[\frac{1}{50000f(n)\log(n)},1 \right]$ for every edge $e \in E^{good}_C$ such that for every vertex $v \in V$,

\begin{equation}
    \label{eq:matching}
    \frac{1}{10}\sum_{e \in E^{good}_C(v)} x_e - \frac{1}{1000f(n)} \leq \sum_{e \in E^{good}_C(v)} x^{intra}_e \leq \frac{1}{2}\sum_{e \in E^{good}_C(v)} x_e + \frac{1}{1000f(n)}.
\end{equation}

We first use the upper bound in \cref{eq:matching} together with the fact that each node $v \in V^{good}$ has a clustering degree of at most $f(n)$ to deduce that $x^{intra}$ is a valid fractional matching.

\begin{claim}
\label{claim:intra_matching_is_good}
For every vertex $v \in V^{good}$, we have $\sum_{e \in E^{good}(v)}x^{intra}_e \leq 1$.
\end{claim}
\begin{proof}
Let $v \in V^{good}$ be arbitrary. We have

\begin{align*}
\sum_{e \in E^{good}(v)} x^{intra}_e &= \sum_{C \in \fC \colon E^{good}_C(v) \neq \emptyset}\sum_{e \in E^{good}_C(v)} x^{intra}_e \\
&\leq \sum_{C \in \fC \colon E^{good}_C(v) \neq \emptyset} \left( \frac{1}{2}\sum_{e \in E^{good}_C(v)} x_e + \frac{1}{1000f(n)}\right) \\
&\leq \frac{1}{2}\sum_{e \in E^{good}(v)} x_e + \frac{\deg_{\fC}(v)}{1000f(n)} \\
&\leq \frac{1}{2} + \frac{1}{1000}  \leq 1,
\end{align*}

as needed.

\end{proof}

Next, we use the lower bound of \cref{eq:matching} together with the fact that $\sum_{e \in E^{good}} x_e \geq 0.8\frac{1}{5}|M^*| \geq \frac{1}{10}|M^*|$ to deduce that $\sum_{e \in E^{good}} x^{intra}_e \geq \frac{1}{40000}|M^*|$.

\begin{claim}
\label{claim:matching_size}
    We have $\sum_{e \in E^{good}} x^{intra}_e \geq \frac{1}{40000}|M^*|$. 
\end{claim}
\begin{proof}
Consider some arbitrary vertex $v \in V^{good}$ with $y^{good}_v := \sum_{e \in E^{good}(v)} x_e \geq \frac{1}{50}$. We have

\begin{align*}
y^{intra}_v &:= \sum_{e \in E^{good}(v)} x^{intra}_e = \sum_{C \in \fC \colon E^{good}_C(v) \neq \emptyset}\sum_{e \in E^{good}_C(v)} x^{intra}_e \\
&\geq \sum_{C \in \fC \colon E^{good}_C(v) \neq \emptyset} \left( \frac{1}{10}\sum_{e \in E^{good}_C(v)} x_e - \frac{1}{1000f(n)}\right) \\
&\geq \frac{1}{10}\sum_{e \in E^{good}(v)} x_e - \frac{\deg_{\fC}(v)}{1000f(n)} \\
&\geq \frac{1}{500} - \frac{1}{1000} 
= \frac{1}{1000}.
\end{align*}
Therefore,

\[\sum_{e \in E^{good}} x^{intra}_e = \frac{1}{2}\sum_{v \in V^{good}} y^{intra}_v \geq \frac{1}{2}\sum_{v \in V^{good} \colon y^{good}_v \geq \frac{1}{50}} y^{intra}_v \geq \frac{1}{2000}|v \in V^{good} \colon y^{good}_v \geq \frac{1}{50}|.\]

Let $E^{loose}$ consists of all edges $\{u,v\}$ in $E^{good}$ with $\max(y^{good}_u,y^{good}_v) < \frac{1}{50}$.
Note that multiplying the fractional value $x_e$ assigned to each edge $e$ in $E^{loose}$ by a factor of $50$ would still result in a valid fractional matching. The value of the fractional matching would be at least $\sum_{e \in E^{loose}} 50x_e$. On the other hand, it is well-known that the value of any fractional matching is at most by a $(3/2)$-factor larger than the size of the maximum matching. Therefore, we get

\[\sum_{e \in E^{loose}} 50x_e \leq (3/2)|M^*|.\]

Thus,

\[|v \in V^{good} \colon y^{good}_v \geq \frac{1}{50}| \geq \sum_{v \in V^{good} \colon y^{good}_v \geq \frac{1}{50}} y_v  \geq \sum_{e \in E^{good}} x_e - \sum_{e \in E^{loose}} x_e \geq \frac{|M^*|}{10} - \frac{(3/2)|M^*|}{50}.\]

Hence, we can conclude that

\[\sum_{e \in E^{good}} x^{intra}_e  \geq \frac{1}{2000}|v \in V^{good} \colon y^{good}_v \geq \frac{1}{50}| \geq \frac{1}{2000} \left(\frac{|M^*|}{10} - \frac{(3/2)|M^*|}{50}\right) \geq \frac{|M^*|}{40000}.\]

\end{proof}

\paragraph{Local Rounding}
Let $E'$ consists of all edges $e$ in $E^{good}$ with $x^{intra}_e > 0$, let $G' := (V,E')$ and $\Delta'$ the maximum degree of $G'$. Note that $x^{intra}_e \geq \frac{1}{50000f(n)\log(n)}$ for every $e \in E'$ together with $\sum_{e \in E'(v)} x^{intra}_e \leq 1$ for every $ v \in V$ directly implies that $\Delta' \leq 50000f(n)\log(n) = 2^{\tilde{O}\log^{2/3}(n)}$. Therefore, in $O(\log^2(\Delta') + \log^* n) = \tilde{O}(\log^{4/3} n)$ rounds we can compute a $3$-approximate maximum matching $M'$ in $G'$ using the algorithm of Fischer \cite{fischer2020improved}. 
Using the fact that the weight of any fractional matching is at most a factor of $(3/2)$ larger compared to the size of the largest integral matching, we get

\[|M'| \geq \frac{1}{3} \left((2/3)\sum_{e \in E'} x^{intra}_e \right) = \frac{2}{9}\sum_{e \in E^{good}}x^{intra}_e \geq \frac{1}{100000}|M^*|.\]

Therefore, $M'$ is a constant approximate matching of $G$, which finishes the proof of \cref{thm:matching_approximate}.
}

\fullOnly{
\subsection{MIS}
In this section, we describe a deterministic distributed algorithm that computes an MIS in $\tilde{O}(\log^2 n)$ rounds of the \LOCAL model. 
}

\mis*

The rest of the section is devoted to the proof of \cref{thm:mis}. \shortOnly{The missing proofs are deferred to the full version.}
We first recall Luby's classic algorithm\cite{luby86}, which in each iteration chooses an independent set of nodes such that, when we add them to the output and remove them from the graph along with their neighbors, in expectation a constant fraction of the edges of the graph are removed.

\paragraph{Luby's Randomized MIS Algorithm.} The starting point is to recall Luby's classic randomized algorithm from \cite{luby1993removing}. 
Each iteration of it works as follows. 
We mark each node $u$ with probability $1/(10 \deg(u))$. Then, for each edge $\{u,v\}$, let us orient the edge as $u\rightarrow v$ if and only if $\deg(u)<\deg(v)$ or $\deg(u)=\deg(v)$ and $ID(u)<ID(v)$. For each marked node $u$, we add $u$ to the independent set if and only if $v$ does not have a marked out-neighbor. Finally, as a clean-up step at the end of this iteration, we remove all nodes that have been added to the independent set along with their neighbors. We then proceed to the next iteration.

\paragraph{Derandomizing Luby via Rounding.} It is well-known that in each iteration of Luby's algorithm a constant fraction of the edges of the remaining graph gets removed, in expectation. Hence, the process terminates in $O(\log n)$ iterations with probability $1 - 1/\poly(n)$. 
We explain how to derandomize each iteration of the algorithm in $\tilde{O}(\log n)$ rounds, such that we still remove a constant fraction of the edges per iteration. 
For the rest of this proof, we focus on an arbitrary iteration, and we assume that $H =(V_H, E_H)$ is the graph induced by the remaining vertices at the beginning of this iteration. Let $\vec{x} \in \{0, 1\}^{V_H}$ be the indicator vector of whether different nodes are marked, that is, we have $x_v=1$ if $v$ is marked and $x_v=0$ otherwise. Let $R_v(\vec{x})$ be the indicator variable of the event that $v$ gets removed, for the marking vector $\vec{x}$. Let $Z(\vec{x})$ be the corresponding number of removed edges. Luby's algorithm determines the markings $\vec{x}$ randomly. Our task is to derandomize this and select the marked nodes in a deterministic way such that when we remove nodes added to the independent set (those marked nodes that do not have a marked out-neighbor) and their neighbors, along with all the edges incident on these nodes, at least a constant fraction of edges $E$ get removed.

Below, we give a pairwise analysis which shows that the expected number of removed edges is $\Omega(|E_H|)$. 
The rounding framework of Faour et al. \cite{faour2022local} would then allow us to select the marked nodes in a deterministic way while retaining the guarantee that $\Omega(|E_H|)$ edges are removed. The rounding procedure runs in $O(\log^2(1/p_{min})) = O(\log^2 \Delta)$ rounds where $p_{min}$ is the smallest marking probability. In fact, this is how Faour et al. \cite{faour2022local} obtained an MIS algorithm running in $O(\log n \log^2 \Delta)$ rounds. 
We derandomize the marking process in $\tilde{O}(\log n)$ rounds.
To do so, we first perform an intra-cluster rounding step and only then apply the rounding framework of Faour et al. \cite{faour2022local}.
In the intra-cluster rounding step, we compute for each node a marking probability which is either $0$ or $1/2^{\tilde{O}(\sqrt{\log n})}$ and such that the expected number of removed edges is still $\Omega(|E_H|)$.
Then, we apply the rounding framework of Faour et al. to the same pairwise analysis as before, but this time with the new marking probabilities. As all non-zero marking probabilities are $1/2^{\tilde{O}(\sqrt{\log n})}$, the rounding procedure runs in just $\left(\log(2^{\tilde{O}(\sqrt{\log n})})\right)^2 = \tilde{O}(\log n)$ rounds, as desired.

\paragraph{Good and bad nodes and prevalence of edges incident on good nodes.} We call any node $v \in V_H$ \emph{good} if and only if it has at least $\deg(v)/3$ incoming edges. A node $v$ that is not good is called bad. It can be proven \cite{luby1993removing} that 
\begin{align}
    \label{eq:luby0}
\sum_{\textit{good vertex\;} v} \deg(v) \geq |E_H|/2.
\end{align}

\fullOnly{Even though the reader may skip this paragraph, for completeness we include the reason as it is a simple and intuitive charging argument. 
Recall that by definition any bad node has less than $1/3$ of its edges incoming. 
Thus any edge incoming to a bad node $v$ can be charged to two unique edges going out of $v$, in such a manner that each edge of the graph is charged at most once. 
Hence, the number of edges incoming to bad nodes is at most $|E_H|/2$. 
Thus, the number of edges that have at least one good endpoint is at least $|E_H|/2$, which implies the desired bound $\sum_{\textit{good vertex\;} v} \deg(v) \geq |E_H|/2$.}

\paragraph{Lower bounding removed edges.}
We can lower bound the number of removed edges as $$Z(\vec{x}) \geq \sum_{\textit{good vertex\;} v} \deg(v) \cdot R_{v}(\vec{x})/2.$$ 
The $2$ factor in the denominator is because for an edge, both endpoints might be good nodes. 
Since $\sum_{\textit{good vertex\;} v} \deg(v) \geq |E_H|/2$, to prove that $\mathbb{E}[Z(\vec{x})] = \Omega(|E_H|)$, it suffices to show that each good vertex $v$ has $Pr[R_v(\vec{x})] =\Omega(1)$. 
This fact can be proven via elementary probability calculations. 
Next, we discuss how to prove it using only pairwise independence in the analysis.

\paragraph{Pessimistic estimator of removed edges via pairwise-independent analysis.} Let us use $IN(u)$ and $OUT(u)$ to denote in-neighbors and out-neighbors of a vertex $u$. Consider a good node $v$ and consider all its incoming neighbors $u$, i.e., neighbors $u$ such that $(\deg(u), ID(u))<(\deg(v), ID(v))$. Since $v$ is good, it has at least $\deg(v)/3$ such neighbors. 
Hence, we have 

\[\sum_{\textit{incoming neighbor}\; u} \frac{1}{\deg(u)} \geq 1/3.\] 
Choose a subset $IN^*(v)\subseteq IN(v)$ of incoming neighbors such that 
\begin{align}
\label{eq:luby1}
\sum_{u \in IN^{*}(v)} \frac{1}{\deg(u)} \in [1/3, 4/3].
\end{align}
 
Notice that such a subset $IN^*(v)$ exists since the summation over all incoming neighbors is at least $1/3$ and each neighbor contributes at most $1$ to the summation. 
On the other hand, notice that for any node $u$, we have 
\begin{align}
\label{eq:luby2}
\sum_{w \in OUT(u)} \frac{1}{\deg(w)} \leq 1.    
\end{align}
This is because $|OUT(u)| \leq \deg(u)$ and for each ${w}\in OUT(u)$, we have $(\deg({w}), ID({w}))>(\deg(u), ID(u))$.

A sufficient event $\mathcal{E}(v,u)$ that causes $v$ to be removed is if some $u\in IN^*(v)$ is marked and no other node in $IN^{*}(v)\cup OUT(u)$ is marked. 
By union bound, this event's indicator is lower bounded by $$x_u - \sum_{u' \in IN^*(v), u\neq u'} x_{u}\cdot x_{u'} - \sum_{w \in OUT(u)} x_{u}\cdot x_{w}.$$ 
Furthermore, the events $\fE(v,u_1), \fE(v,u_2), \dots, \fE(v,u_{|IN^*(v)|})$ are mutually disjoint for different $u_1, u_2, \dots,$ $u_{|IN^*(v)|} \in IN^*(v)$. 
Hence, we can sum over these events for different $u\in IN^*(v)$ and conclude that

\fullOnly{
\begin{align*}
R_v(\vec{x})  \geq \sum_{u\in IN^*(v)} \bigg(x_u - \sum_{u' \in IN^*(v), u\neq u'} x_{u}\cdot x_{u'} - \sum_{w \in OUT(u)} x_{u}\cdot x_{w}\bigg)  \\
  =
\sum_{u\in IN^*(v)} x_u - \sum_{u, u' \in IN^*(v)} x_{u}\cdot x_{u'}
- \sum_{u\in IN^*(v)}\sum_{w \in OUT(u)} x_{u}\cdot x_{w}
\end{align*}
}
\shortOnly{
\begin{align*}
&R_v(\vec{x})  \\
\geq &\sum_{u\in IN^*(v)} \bigg(x_u - \sum_{u' \in IN^*(v), u\neq u'} x_{u}\cdot x_{u'} - \sum_{w \in OUT(u)} x_{u}\cdot x_{w}\bigg)  \\
  =
&\sum_{u\in IN^*(v)} x_u - \sum_{u, u' \in IN^*(v)} x_{u}\cdot x_{u'} \\
& \;\;\;\;\;\;\;\;\;\;\;\;\;\;\;\;\; - \sum_{u\in IN^*(v)}\sum_{w \in OUT(u)} x_{u}\cdot x_{w}
\end{align*}
}
Therefore, our overall pessimistic estimator for the number of removed edges gives that 

\fullOnly{
\begin{align*}
Z(\vec{x}) \geq & \sum_{\textit{good vertex \,} v}  (\deg(v)/2) \cdot R_v(\vec{x})  \\
\geq  & \sum_{\textit{good vertex \,} v} (\deg(v)/2) \cdot \bigg(
\sum_{u\in IN^*(v)} x_u - \sum_{u, u' \in IN^*(v)} x_{u}\cdot x_{u'} - \sum_{u\in IN^*(v)}\sum_{w \in OUT(u)} x_{u}\cdot x_{w}\bigg).
\end{align*}
}\shortOnly{
\begin{align*}
Z(\vec{x}) \geq &\sum_{\textit{good vertex \,} v}  (\deg(v)/2) \cdot R_v(\vec{x}) &  \\
\geq   &\sum_{\textit{good vertex \,} v} (\deg(v)/2) \cdot \bigg(
\sum_{u\in IN^*(v)} x_u - \sum_{u, u' \in IN^*(v)} x_{u}\cdot x_{u'} \\
& \;\;\;\;\;\;\;\;\;\;\;\;\;\;\;\;\;\;\;\;\;\;\;\;\;\;\;\;\;\;\;\;\;\;\;\;\;\;\;\; -  \sum_{u\in IN^*(v)}\sum_{w \in OUT(u)} x_{u}\cdot x_{w}\bigg).
\end{align*}
}
\paragraph{Intra-Cluster Rounding}

For the intra-cluster rounding step, we assume that we are given a partition $\fC_G$ of the input graph $G$ with diameter $O(\sqrt{\log n})$ and such that $\deg_{\fC_G}(u) \leq f(n)$ for every $u \in V$ and for some function $f(n) = 2^{\tilde{O}(\sqrt{\log n})}$.

Indeed, invoking \cref{thm:clustering_all} with $\alpha = \sqrt{\log(N)}$, we can compute such a partition $\fC_G$ at the very beginning in $\tilde{O}(\log^2 n)$ rounds.
Also, we denote by $\fC_H = \{C \cap V_H \colon C \in \fC_G\}$ the partition of $H$ that we obtain from the partition $\fC_G$ by removing from each cluster the vertices not contained in $V_H$.

Now, consider that we relax the label assignment such that it also allows for a fractional assignment $\vec{x} \in [0,1]^{V_H}$. Then, $Z(\vec{x})$ is a pessimistic estimator on the expected number of edges removed if we mark each vertex $u$ fully independently (or pairwise independent) with probability $x_u$.
In Luby's algorithm, one marks each node $u \in V_H$ with probability $x_u := \frac{1}{20\deg(u)}$, and a simple calculation shows that for this marking probability we get $Z(\vec{x}) = \Omega(|E_H|)$.
In the intra-cluster rounding step, the goal is to compute a $\frac{1}{10000f(n)\log(n)}$-integral assignment $\vec{x}^{intra} \in [0,1]^{V_H}$ such that we have $Z(\vec{x}^{intra}) = \Omega(|E_H|)$.

Consider some fixed cluster $C \in \fC_H$. For any $u \in C$, we define $\vec{x}^{intra}$ in such a way that we can compute $x^{intra}_u$ by only knowing the cluster $C$ (together with the $k$-hop neighborhood around $C$ for some $k = O(1)$).
As the weak-diameter of $C$ in the original input graph is $O(\sqrt{\log n})$, we can compute $x^{intra}_u$ in $O(\sqrt{\log n})$ rounds.
Similar as in the intra-cluster rounding step for constant approximate matching, we will use the probabilistic method to show that a fractional assignment with certain desirable properties exist.

To that end, we introduce one random variable $X^{intra}_u$ for every vertex $u \in C$. If $\deg(u) \leq 1000f(n)\log(n)$, we simply set $X^{intra}_u = \frac{1}{10\deg(u)}$. Otherwise, if $\deg(u) > 1000f(n) \log(n)$, we set $X^{intra}_u = \frac{1}{10000f(n) \log(n)}$ with probability $\frac{1000f(n)\log(n)}{\deg(u)}$ and with the remaining probability we set $X^{intra}_u = 0$, fully independently. Note that for every $u \in C$, $\E[X_u] = \frac{1}{10\deg(u)}$.

The following claim follows by a simple Chernoff bound.

\begin{claim}
    \label{claim:mis_two}
    For any subset $S \subseteq C$, we have

    \[Pr \left[ \sum_{u \in S} X^{intra}_u \leq  \frac{1}{5}\sum_{u \in S} \frac{1}{\deg(u)} + \frac{1}{100f(n)} \right] \geq 1 - \frac{1}{n^{10}} \text{ and}\]
    \[Pr \left[\sum_{u \in S} X^{intra}_u \geq \frac{1}{20}\sum_{u \in S}\frac{1}{\deg(u)} - \frac{1}{100f(n)}\right] \geq 1 - \frac{1}{n^{10}}.\]
\end{claim}

\fullOnly{
\begin{proof}
    Let $S^{large}$ consists of all nodes $u \in S$ with $\deg(u) > 1000f(n)\log(n)$ and let $X := \sum_{u \in S^{large}} X^{intra}_u$. Note that $X$ is the sum of independent random variables taking values between $0$ and $b:= \frac{1}{10000f(n)\log(n)}$ and $\E[X] = \sum_{u \in S^{large}} \frac{1}{10\deg(u)}$. Let $t := \frac{1}{20}\sum_{u \in S^{large}} \frac{1}{\deg(u)} +\frac{1}{100f(n)}$. We have $t \geq 0.5\E[X]$ and therefore we can use  the Chernoff bound variant given in \cref{cor:chernoff} to deduce
    
    \[Pr \left[ |X-\E[X]| \geq t \right] \leq 2e^{-\frac{t}{6b}} \leq n^{-10}.\]
    Now, assume that $|X-\E[X]| \leq t$.
    Then, 
    \[\sum_{u \in S}X^{intra}_u = \sum_{u \in S \setminus S^{large}}X^{intra}_u + X \leq \sum_{u \in S \setminus S^{large}} \frac{1}{10\deg(u)} + \E[X] + t = \sum_{u \in S} \frac{1}{10\deg(u)}  + t \leq \sum_{u \in S}\frac{1}{5\deg(u)} + \frac{1}{100f(n)}.\]

    Similarly,

     \[\sum_{u \in S}X^{intra}_u \geq \sum_{u \in S \setminus S^{large}} \frac{1}{10\deg(u)} + \E[X] - t = \sum_{u \in S} \frac{1}{10\deg(u)}  - t \geq \sum_{u \in S}\frac{1}{20\deg(u)} - \frac{1}{100f(n)}.\]
\end{proof}
}

\cref{claim:mis_two} together with a simple union bound implies that the following holds with strictly positive probability: for every good vertex $v$,

\fullOnly{
\[\frac{1}{20}\sum_{u \in IN^*(v) \cap C} \frac{1}{\deg(u)} - \frac{1}{100f(n)} \leq \sum_{u \in IN^*(v) \cap C} X^{intra}_u \leq \frac{1}{5}\sum_{u \in IN^*(v) \cap C} \frac{1}{\deg(u)} + \frac{1}{100f(n)}\]
}
\shortOnly{
\begin{align*}
   \frac{1}{20}\sum_{u \in IN^*(v) \cap C} \frac{1}{\deg(u)} - \frac{1}{100f(n)} \leq \sum_{u \in IN^*(v) \cap C} X^{intra}_u \\
   \leq \frac{1}{5}\sum_{u \in IN^*(v) \cap C} \frac{1}{\deg(u)} + \frac{1}{100f(n)} 
\end{align*}

}
and for every vertex $u \in V_H$,

\[\sum_{w \in OUT(u) \cap C} X^{intra}_w \leq \frac{1}{5}\sum_{w \in OUT(u) \cap C} \frac{1}{\deg(w)} + \frac{1}{100f(n)}.\]

Therefore, each cluster $C$ can compute in $O(\sqrt{\log n})$ rounds a value $x^{intra}_u \in \{0\} \cup \left[\frac{1}{10000f(n)\log(n)},1 \right]$ for every node $u \in C$ such that for every good vertex $v$,
\fullOnly{
\[\frac{1}{20}\sum_{u \in IN^*(v) \cap C} \frac{1}{\deg(u)} - \frac{1}{100f(n)} \leq \sum_{u \in IN^*(v) \cap C} x^{intra}_u \leq \frac{1}{5}\sum_{u \in IN^*(v) \cap C} \frac{1}{\deg(u)} + \frac{1}{100f(n)}\]
}
\shortOnly{
\begin{align*}
 \frac{1}{20}\sum_{u \in IN^*(v) \cap C} \frac{1}{\deg(u)} - \frac{1}{100f(n)} \leq \sum_{u \in IN^*(v) \cap C} x^{intra}_u \\ \leq \frac{1}{5}\sum_{u \in IN^*(v) \cap C} \frac{1}{\deg(u)} + \frac{1}{100f(n)}   
\end{align*}
}
and for every vertex $u \in V_H$,

\[\sum_{w \in OUT(u) \cap C} x^{intra}_w \leq \frac{1}{5}\sum_{w \in OUT(u) \cap C} \frac{1}{\deg(w)} + \frac{1}{100f(n)}.\]

We use these properties together with the fact that the degree of the clustering $\fC_H$ is at most $f(n)$ to prove the claim below.

\begin{claim}
\label{claim:mis_inoutsmall}
For every good vertex $v$, we have $$1/1000 \leq \sum_{u \in IN^*(v)} x^{intra}_u \leq 1/3.$$ For every $u \in V_H$, we have $\sum_{w \in OUT(u)} x^{intra}_w \leq 1/4$.
\end{claim}

\fullOnly{
\begin{proof}
Let $v$ be an arbitrary good vertex. We have

\begin{align*}
\sum_{u \in IN^*(v)} x^{intra}_u &= \sum_{C \in \fC_H \colon IN^*(v)\cap C \neq \emptyset}\sum_{u \in IN^*(v) \cap C} x^{intra}_u \\
&\leq \sum_{C \in \fC_H  \colon IN^*(v)\cap C \neq \emptyset} \left( \frac{1}{5}\sum_{u \in IN^*(v) \cap C} \frac{1}{\deg(u)} + \frac{1}{100f(n)}\right) \\
&\leq \frac{1}{5}\sum_{u \in IN^*(v)} \frac{1}{\deg(u)} + \frac{\deg_{\fC_H}(v)}{100f(n)} \\
&\leq \frac{1}{5}\frac{4}{3} + \frac{1}{100} \\
&\leq 1/3,
\end{align*}

where we used \cref{eq:luby1}.

Similarly, 

\begin{align*}
\sum_{u \in IN^*(v)} x^{intra}_u &= \sum_{C \in \fC_H \colon IN^*(v)\cap C \neq \emptyset}\sum_{u \in IN^*(v) \cap C} x^{intra}_u \\
&\geq \sum_{C \in \fC_H  \colon IN^*(v)\cap C \neq \emptyset} \left( \frac{1}{20}\sum_{u \in IN^*(v) \cap C} \frac{1}{\deg(u)} - \frac{1}{100f(n)}\right) \\
&\geq \frac{1}{20}\sum_{u \in IN^*(v)} \frac{1}{\deg(u)} - \frac{\deg_{\fC_H}(v)}{100f(n)} \\
&\geq \frac{1}{20}\frac{1}{3} - \frac{1}{100} \\
&\geq 1/1000,
\end{align*}

where we used \cref{eq:luby1}.

Next, consider an arbitrary $u \in V_H$. In a similar fashion as above, we have

\begin{align*}
\sum_{w \in OUT(u)} x^{intra}_w &= \sum_{C \in \fC_H \colon OUT(u)\cap C \neq \emptyset}\sum_{w \in OUT(u) \cap C} x^{intra}_w \\
&\leq \sum_{C \in \fC_H  \colon OUT(u) \cap C \neq \emptyset} \left(\frac{1}{5}\sum_{w \in OUT(u) \cap C} \frac{1}{\deg(w)} + \frac{1}{100f(n)}\right) \\
&\leq \frac{1}{5}\sum_{w \in OUT(u)} \frac{1}{\deg(w)} + \frac{\deg_{\fC_H}(u)}{100f(n)} \\
&\leq \frac{1}{5} + \frac{1}{100} \\
&\leq 1/4,
\end{align*}
where we used \cref{eq:luby2}.
\end{proof}
}

\paragraph{Local Rounding}

We next round the fractional solution $\vec{x}^{intra}$ to an integral solution $\vec{y}$ using the rounding framework of Faour et al. \cite{faour2022local}, as discussed in \Cref{subsec:roundingFramework}.

For a given label assignment $\vec{x} \in \{0,1\}^{V_H}$, we define the utility function as $$\utility(\vec{x}) = \sum_{\textit{good vertex \,} v} (\deg(v)/2) \cdot \big(\sum_{u\in IN^*(v)} x_u\big),$$ and the cost function as 
\fullOnly{
$$\cost(\vec{x}) =  \sum_{\textit{good vertex \,} v} (\deg(v)/2) \cdot \bigg(
\sum_{u, u' \in IN^*(v)} x_{u}\cdot x_{u'} + \sum_{u\in IN^*(v)}\sum_{w \in OUT(u)} x_{u}\cdot x_{w}\bigg).$$
}
\shortOnly{
\begin{align*}
  \cost(\vec{x}) =  \sum_{\textit{good vertex \,} v} (\deg(v)/2) \cdot \bigg(
\sum_{u, u' \in IN^*(v)} x_{u}\cdot x_{u'}  \\ +  \sum_{u\in IN^*(v)}\sum_{w \in OUT(u)} x_{u}\cdot x_{w}\bigg).  
\end{align*}

}

If the label assignment is relaxed to be a fractional assignment $\vec{x}\in [0,1]^{V_H}$, where intuitively now $x_u$ is the probability of $u$ being marked, the same definitions apply for the utility and cost of this fractional assignment. 
    
    Let $H^2$ denote the graph where any two nodes of distance at most $2$ in $H$ are connected by an edge.
    Note that $\utility(\vec{x})$ is a utility function in the graph $H^2$ and similarly $\cost(\vec{x})$ is a cost function in the graph $H^2$.
We next argue that the utility and cost function also satisfy the key requirement of \Cref{lemma:rounding}:

\begin{claim}\label[claim]{clm:MIS1} For the fractional label assignment $\vec{x}^{intra}\in [0,1]^{V_H}$ computed during the intra-cluster rounding step we have $\utility(\vec{x}^{intra})-\cost(\vec{x}^{intra}) \geq \utility(\vec{x}^{intra})/3$.
\end{claim}
\fullOnly{
\begin{proof}
We have \begin{align*}
&\utility(\vec{x}^{intra})-\cost(\vec{x}^{intra})  \\
=& \sum_{\textit{good vertex \,} v} (\deg(v)/2) \cdot \bigg(\sum_{u\in IN^*(v)} x^{intra}_u - \sum_{u, u' \in IN^*(v)} x^{intra}_{u}\cdot x^{intra}_{u'} - \sum_{u\in IN^*(v)}\sum_{w \in OUT(u)} x^{intra}_{u}\cdot x^{intra}_{w}\bigg)  \\
= & \sum_{\textit{good vertex \,} v} (\deg(v)/2) \cdot \bigg(\sum_{u\in IN^*(v)} x^{intra}_u \cdot \big(1- \sum_{u' \in IN^*(v)} x^{intra}_{u'} - \sum_{w \in OUT(u)} x^{intra}_{w}\big)\bigg)  \\
\geq & \sum_{\textit{good vertex \,} v} (\deg(v)/2) \cdot \bigg(\sum_{u\in IN^*(v)} x^{intra}_u \big(1- 1/3 - 1/3 \big)\bigg) \\
\geq & \sum_{\textit{good vertex \,} v} (\deg(v)/2) \cdot \big(\sum_{u\in IN^*(v)} x^{intra}_u/3\big)  = \utility(\vec{x})/3,
\end{align*}
where we used \cref{claim:mis_inoutsmall}.
\end{proof}
}

Hence, we can apply \Cref{lemma:rounding} on these fractional assignments with $\lambda_{min} = \frac{1}{100f(n)\log(n)} = 1/2^{\tilde{O}(\sqrt{\log n})}$. The algorithm runs in  $O(\log^2(1/\lambda_{min}) + \log(1/\lambda_{min})\log^*n) = \tilde{O}(\log n)$ rounds in $H^2$, and hence can be simulated with no asymptotic overhead in $H$, and as a result we get an integral label assignment $\vec{y} \in \{0,1\}^{V_H}$ which satisfies $\utility(\vec{y}) - \cost(\vec{y}) \geq 0.5 (\utility(\vec{x}^{intra}) - \cost(\vec{x}^{intra}))$. We know that $Z(\vec{y})=\utility(\vec{y})-\cost(\vec{y}) \geq (1/2) \cdot (\utility(\vec{x}^{intra})-\cost(\vec{x}^{intra})).$ Next, we argue that this implies $Z(\vec{y}) \geq |E_H|/24000$.

\begin{claim}\label[claim]{clm:MIS2} For the fractional label assignment $\vec{x}^{intra}\in [0,1]^{V_H}$ computed during the intra-cluster rounding step we have $Z(\vec{x}^{intra}) = \utility(\vec{x}^{intra})-\cost(\vec{x}^{intra}) \geq |E_H| /12000$. Hence, for the integral marking assignment $\vec{y}$ we obtain from rounding $\vec{x}$ by invoking \Cref{lemma:rounding}, we have $Z(\vec{y})=\utility(\vec{y})-\cost(\vec{y}) \geq (1/2) \cdot (\utility(\vec{x})-\cost(\vec{x})) \geq |E_H|/24000$.
\end{claim}
\fullOnly{
\begin{proof}From \Cref{clm:MIS1}, we have $Z(\vec{x}^{intra}) =\utility(\vec{x}^{intra})-\cost(\vec{x}^{intra}) \geq \utility(\vec{x}^{intra})/3$. Hence, 

\begin{align*} Z(\vec{x}^{intra}) \geq \utility(\vec{x}^{intra})/3 &= 
\sum_{\textit{good vertex \,} v} (\deg(v)/2) \bigg(\sum_{u\in IN^*(v)} x^{intra}_u/3\bigg)  \\
&\ge \sum_{\textit{good vertex \,} v} (\deg(v)/2) \cdot (1/3000) \geq |E_H|/12000,
\end{align*}
where we first used \cref{claim:mis_inoutsmall} that says that $\sum_{u \in IN^*(v)} x^{intra}_u \ge 1/1000$ and then we used \cref{eq:luby0} that bounds $\sum_{\textit{good vertex \,} v} \deg(v) \geq |E_H|/2$. 

Since $Z(\vec{y})=\utility(\vec{y})-\cost(\vec{y}) \geq (1/2) \cdot (\utility(\vec{x})-\cost(\vec{x}))$, the claim follows. 
\end{proof}
}
\paragraph{Putting Everything Together}

From the rounding procedure described above, which runs in $\tilde{O}(\log n)$ rounds of the \LOCAL model,  we get an integral marking assignment $\vec{y}$ with the following guarantee: if we add marked nodes $u$ that have no marked out-neighbor to the independent set and remove them along with their neighbors, we remove at least a $1/24000$ fraction of the remaining edges. Hence, $O(\log n)$ such iterations suffice to complete the computation and have a maximal independent set, for a total round complexity of $\tilde{O}(\log^2 n)$. Note that the low-diameter partition only has to be computed once in the beginning, which takes $\tilde{O}(\log^ 2 n)$ rounds. Thus, we can indeed compute an MIS in $\tilde{O}(\log^2 n)$ rounds of the \local model. \shortOnly{\balance}


\fullOnly{
\bibliographystyle{alpha}
\bibliography{ref}
}

\fullOnly{
\appendix

\section{The deterministic hitting set algorithm of \Cref{lem:hitting_set_good_fraction}}\label{app:hittingSet} 
In this section, we present the proof of \Cref{lem:hitting_set_good_fraction}. This is the deterministic hitting set subroutine used in our clustering results. For that, in \Cref{subsec:basicHittingSet}, we first develop a more basic variant of the hitting set result. Then, in \Cref{subsec:hittingGoodFraction}, we explain how to go from this to \Cref{lem:hitting_set_good_fraction}.

\subsection{A basic hitting set algorithm} 
\label{subsec:basicHittingSet}
\begin{lemma}
   \label{lem:hitting_set}
   There is a deterministic distributed algorithm in the \local model that, for every $\Delta \in \mathbb{N}$, $b \in \mathbb{N}$, $p = \Omega(1/\Delta)$ and $Norm \geq 0$, it provides the following guarantees:
   The input is a bipartite graph $H = (U_H \sqcup  V_H, E_H)$ with  $\deg_H(u) = \Delta$ for every $u \in U_H$. Initially, each node in $H$ is equipped with a unique $b$-bit identifier and each node $u \in U_H$ is assigned a weight $w_u \geq 0$. Each node also knows at the beginning to which side of the bipartition it belongs. The algorithm computes a subset $V^{sub} \subseteq V_H$ satisfying
   
   \[\sum_{u \in U_H \colon N_H(u) \cap V^{sub} = \emptyset}w_u + Norm \cdot |V^{sub}| \leq e^{-p\Delta} \sum_{u \in U_H} w_u + Norm \cdot 4p \cdot |V_H|.\]
   The algorithm runs in $O(\Delta p (\log^2(\Delta) + \log(\Delta)\log^* b))$ rounds.
\end{lemma}

The rest of this subsection is dedicated to proving \Cref{subsec:basicHittingSet}. 

The algorithm runs in $T := \lceil 10 p \Delta \rceil$ steps. For $i \in [1,T]$, the algorithm computes a set $S_i \subseteq V_H$. For $i \in [0,T]$, we define $W_i = \cup_{j \in [1,i]} S_j$ and $U^{unhit}_i = \{u \in U_H \colon N_H(u) \cap W_i = \emptyset\}$. In the end, the algorithm outputs $V^{sub} := W_T$.
For $i \in [0,T]$, we define 

\[\Phi_i = e^{-\frac{T-i}{T}p\Delta} \sum_{u \in U^{unhit}_i} w_u  + Norm \cdot |W_i| + \frac{T-i}{T}Norm \cdot 4p \cdot |V_H|.\]

Note that

\[\Phi_0 = e^{-p\Delta}\sum_{u \in U^{unhit}_0} w_u + Norm \cdot 4p \cdot |V_H| = e^{-p\Delta}\sum_{u \in U_H} w_u + Norm \cdot 4p \cdot |V_H|\]

and

\[\Phi_T = \sum_{u \in U^{unhit}_T} w_u + Norm \cdot |W_T| = \sum_{u \in U_H \colon N_H(u) \cap V^{sub} = \emptyset} w_u + Norm \cdot |V^{sub}|.\]

Therefore, it suffices to ensure that $\Phi_T \leq \Phi_0$. In fact, the algorithm computes each $S_i$ in a way that ensures $\Phi_i \leq \Phi_{i-1}$ for every $i \in [1,T]$. We first introduce the notion of a good set $S_i$. On one hand, a good set $S_i$ indeed ensure that $\Phi_i \leq \Phi_{i-1}$. On the other hand, it is defined in such a way that one can apply the local rounding framework of Faour et al. \cite{faour2022local} in a straightforward manner to efficiently compute a good set.

\begin{definition}[Good Set $S_i$ in Step $i$]
    \label{def:clustering_good_set}
    Let $i \in [1,T]$ be arbitrary.
    For a set $S_i \subseteq V_H$ and $u \in U_H$, let
    
    \[Y_i(u) = 1 - |N_H(u) \cap S_i| + \binom{|N_H(u) \cap S_i|}{2}.\]
    
    We refer to $S_i$ as good in step $i$ if
    \fullOnly{
    \[e^{-\frac{T-i}{T}p\Delta}\sum_{u \in U^{unhit}_{i-1}} Y_i(u) w_u + Norm \cdot |S_i| \leq  e^{-\frac{T-(i-1)}{T}p\Delta} \sum_{u \in U^{unhit}_{i-1}} w_u  + \frac{Norm \cdot 4p}{T} \cdot |V_H|. \]
    }
    \shortOnly{
    \begin{align*} &&e^{-\frac{T-i}{T}p\Delta}\sum_{u \in U^{unhit}_{i-1}} Y_i(u) w_u + Norm \cdot |S_i|  \\ \leq  &&e^{-\frac{T-(i-1)}{T}p\Delta} \sum_{u \in U^{unhit}_{i-1}} w_u  + \frac{Norm \cdot 4p}{T} \cdot |V_H|. \end{align*}
    }
    
\end{definition}

Note that $Y_i(u) \geq I(N_H(u) \cap S_i = \emptyset)$. Recall that for an event $\mathcal{E}$, we define the indicator variable $I(\mathcal{E})$ to be equal to $1$ if $\mathcal{E}$ happens and $0$ otherwise. 

\begin{claim}
   For $i \in [1,T]$, if $S_i$ is a good set in step $i$, then $\Phi_i \leq \Phi_{i-1}$.
\end{claim}
\begin{proof}
    We have
    \begin{align*}
        \Phi_i &= e^{-\frac{T-i}{T}p\Delta} \sum_{u \in U^{unhit}_i} w_u  + Norm \cdot |W_i| + \frac{T-i}{T}Norm \cdot 4p \cdot |V_H| \\
        &\leq e^{-\frac{T-i}{T}p\Delta} \sum_{u \in U^{unhit}_{i-1}} Y_i(u) w_u  + Norm \cdot |S_i| + Norm \cdot |W_{i-1}| + \frac{T-i}{T}Norm \cdot 4p \cdot |V_H| \\
        &\leq e^{-\frac{T-(i-1)}{T}p\Delta} \sum_{u \in U^{unhit}_{i-1}} w_u  + \frac{Norm \cdot p}{T} \cdot |V_H| + Norm \cdot |W_{i-1}| + \frac{T-i}{T}Norm \cdot 4p \cdot |V_H| \\
        &= \Phi_{i-1},
    \end{align*}
    as needed.
\end{proof}

\begin{lemma}
   For a fixed $i \in [1,T]$, we can compute a good set $S_i$ in step $i$ in $O(\log^2 \Delta + \log(\Delta)\log^* b)$ rounds.
\end{lemma}
\begin{proof}
We make use of the local rounding framework of Faour et al.~\cite{faour2022local} as outlined in \Cref{subsec:roundingFramework} to compute a good set $S_{i}$. To help readability, let us recall the related definition and restate their main rounding lemma.

\DEFutilcost*

\FaouretalRounding*

    \paragraph{Our Local Derandomization.} 
    The labeling space is whether each node in $V_H$ is contained in $S_i$ or not, i.e., each node in $V_H$ takes simply one of two possible labels $\Labels=\{0,1\}$ where $1$ indicates that the node is in $S_i$. For a given label assignment $\vec{x} \in \{0,1\}^{V_H}$, we define the utility function
    
    \[\utility(\vec{x}) = e^{-\frac{T-i}{T}p \Delta}\sum_{u \in U^{unhit}_{i-1}} w_u \sum_{v \in N_H(u)} x_v + \frac{Norm \cdot 4p}{T} \cdot |V_H|,\]
    and the cost 
    \[\cost(\vec{x}) = e^{-\frac{T-i}{T}p \Delta}\sum_{u \in U^{unhit}_{i-1}} w_u  \sum_{v \neq v' \in N_H(u)} x_v x_{v'} + Norm\cdot \sum_{v \in V_H} x_v.\]
    
    If the label assignment is relaxed to be a fractional assignment $\vec{x} \in [0,1]^{V_H}$, where intuitively now $x_v$ is the probability of $v$ being contained in $S_{i}$, the same definitions apply for the utility and cost of this fractional assignment.
    
    Let $G = (V_G,E_G)$ be the graph with $V_G = V_H$ and where any two vertices $v \neq v' \in V_H$ are connected by an edge if $v$ and $v'$ have a common neighbor in $U_H$.

    Note that $\utility(\vec{x})$ is a utility function in the graph $G$ and similarly $\cost(\vec{x})$ is a cost function in the graph $G$.
    
    We next argue that the fractional assignment where $x_v = \frac{2p}{T}$ for each $v\in V_H$ satisfies the conditions of \Cref{lemma:rounding}. For the given fractional assignment, utility minus cost is at least a constant factor of utility.
    
    \begin{claim}
    Let $\vec{x} \in [0,1]^{V_H}$ with $x_v = \frac{2p}{T}$ for every $v \in V_H$. Then, $\utility(\vec{x}) - \cost(\vec{x}) \geq \utility(\vec{x})/2$.
    \end{claim}
    \begin{proof}
    Note that for every $u \in U_H$, we have
    \[\sum_{v \in N_H(u)} x_v = \Delta \cdot \frac{2p}{T} \leq 0.2.\]
    Therefore,
    
    \begin{align*}
         \utility(\vec{x}) &= e^{-\frac{T-i}{T}p \Delta}\sum_{u \in U^{unhit}_{i-1}} w_u  \sum_{v \in N_H(u)} x_v + \frac{Norm \cdot 4p}{T} \cdot |V_H| \\
         &\geq 2 \left(e^{-\frac{T-i}{T}p \Delta}\sum_{u \in U^{unhit}_{i-1}} w_u \sum_{v \neq v' \in N_H(u)} x_v x_{v'} + Norm \cdot \sum_{v \in V_H} x_v\right) \\
         &\geq 2 \cost(\vec{x})
    \end{align*}

    and thus indeed $\utility(\vec{x}) - \cost(\vec{x}) \geq \utility(\vec{x})/2$.
    \end{proof}
    
    Hence, we can apply \Cref{lemma:rounding} on these fractional assignments with $\lambda_{min} = \frac{2p}{T} = \Omega(1/\Delta)$. The algorithm runs in  $O(\log^2 \Delta + \log(\Delta)\log^*b)$ rounds and as a result we get an integral label assignment 
    $\vec{y} \in \{0,1\}^{V_H}$ which satisfies $\utility(\vec{y}) - \cost(\vec{y}) \geq 0.9 (\utility(\vec{x}) - \cost(\vec{x}))$. We can then conclude
    \begin{align*}
        \utility(\vec{y}) - \cost(\vec{y}) &\geq 0.9 (\utility(\vec{x}) - \cost(\vec{x})) \\
        &\geq0.9 \left( e^{-\frac{T-i}{T}p \Delta}\sum_{u \in U^{unhit}_{i-1}} w_u \left( \Delta \cdot \frac{2p}{T} - \binom{\Delta}{2} \left(\frac{2p}{T}\right)^2\right) + \frac{Norm \cdot 4p}{T} \cdot |V_H| - \frac{Norm \cdot 2p}{T} |V_H| \right)\\
        &\geq e^{-\frac{T-i}{T}p \Delta} \sum_{u \in U^{unhit}_{i-1}} w_u  \frac{\Delta \cdot p}{T}.
    \end{align*}
    This integral label assignment directly gives us $S_i$. In particular, let $S_i= \{v \in V_H \colon y_v = 1\}$. Note that
    \begin{align*}
        \utility(\vec{y}) - \cost(\vec{y}) = e^{-\frac{T-i}{T}p \Delta} \sum_{u \in U^{unhit}_{i-1}} \left( |N_H(u) \cap S_{i,j}| - \binom{|N_H(u)\cap S_{i,j}|}{2} \right)w_u + \frac{Norm \cdot 4p}{T}|V_H| - Norm|S_i|
    \end{align*}
    and therefore

    \begin{align*}
        e^{-\frac{T-i}{T}p\Delta}\sum_{u \in U^{unhit}_{i-1}} Y_i(u) w_u + Norm \cdot |S_i| &= e^{-\frac{T-i}{T}p \Delta} \sum_{u \in U^{unhit}_{i-1}} w_u - \utility(\vec{y}) + \cost(\vec{y}) + \frac{Norm \cdot 4p}{T}|V_H| \\
        &\leq e^{-\frac{T-i}{T}p \Delta} \sum_{u \in U^{unhit}_{i-1}} \left(1 - \frac{\Delta \cdot p}{T}\right)w_u + \frac{Norm \cdot 4p}{T}|V_H| \\
        &\leq e^{-\frac{T-(i-1)}{T}p \Delta} \sum_{u \in U^{unhit}_{i-1}} w_u + \frac{Norm \cdot 4p}{T}|V_H|
    \end{align*}

    which shows that $S_i$ is indeed a good set in step $i$ according to \Cref{def:clustering_good_set}. 
\end{proof}

\subsection{The hitting set algorithm of \Cref{lem:hitting_set_good_fraction}}
\label{subsec:hittingGoodFraction}
We are now ready to prove \Cref{lem:hitting_set_good_fraction}. To help readability, we first restate the lemma.
\hittingSetGoodFraction*

\begin{proof}[Proof of \Cref{lem:hitting_set_good_fraction}]
Let $H' = (U_{H'} \sqcup V_{H'}, E_{H'})$ be the bipartite graph we obtain from $H$ by  replacing each vertex $u \in U_H$ with $\lfloor \frac{\Delta}{k} \rfloor$ copies of it and connecting each copy of $u$ to $k$ neighbors of $u$ in $H$ in such a way that no two copies are connected to the same neighbor. 
Note that given the unique $b$-bit identifiers in $H$, it is easy to assign unique $(2 + \lceil \log_2(\Delta)\rceil) = O(b)$-bit identifiers to nodes in $H'$.

According to \cref{lem:hitting_set}, there exists a \local algorithm, running in $O(k p (\log^2(k) + \log(k)\log^* b))$ rounds, which computes a set $V^{sub} \subseteq V_{H'} = V_H$ satisfying

\begin{align*}
     \sum_{u' \in U_{H'} \colon N_{H'}(u') \cap V^{sub} = \emptyset} \frac{2w_{u'}}{\lfloor \Delta/k\rfloor} + Norm \cdot |V^{sub}| &\leq e^{-pk}\sum_{u' \in U_{H'}} \frac{2w_{u'}}{\lfloor \Delta/k\rfloor} + Norm \cdot 4p \cdot |V_H| \\
     &\leq 4\left(e^{-pk}\sum_{u \in U_{H}} w_u + Norm \cdot p \cdot |V_H| \right),
\end{align*}

where for each copy $u'$ of $u$, we set $w_{u'} := w_u$.
 
 For every node $u \in U_H$ with $|N_H(u) \cap V^{sub}| \leq 0.5 \lfloor \Delta/k\rfloor$, there exist at least $0.5 \lfloor \Delta/k\rfloor$ copies $u'$ of $u$ with $N_{H'}(u') \cap V^{sub} = \emptyset$. Therefore, 
 
 \[\sum_{u \in U_H \colon |N_H(u) \cap V^{sub}| \leq 0.5 \lfloor \Delta/k \rfloor} w_u \leq \sum_{u' \in U_{H'} \colon N_{H'}(u') \cap V^{sub} = \emptyset} \frac{2w_{u'}}{\lfloor \Delta/k\rfloor},\]
 
 which together with the previous calculation shows that $V^{sub}$ fulfills the condition of \cref{lem:hitting_set_good_fraction}.
\end{proof}

}

\end{document}